\documentclass[graybox]{svmult}

\usepackage{type1cm}        
\usepackage{makeidx}         
\usepackage{graphicx}        
\usepackage{multicol}        
\usepackage[bottom]{footmisc}

\usepackage{newtxtext}       %
\usepackage[varvw]{newtxmath}       

\makeindex             

\usepackage{graphicx} 
\usepackage{amsmath,mathtools}
\usepackage{mathrsfs}
\usepackage{nicefrac} 
\usepackage[linesnumbered,ruled,vlined]{algorithm2e}
\usepackage[colorlinks,citecolor=blue,urlcolor=blue,pdfauthor=author]{hyperref}
\usepackage[round]{natbib}

\usepackage{caption}
\usepackage{subcaption}

\usepackage[shortlabels]{enumitem}
\setlist[itemize]{label=\textbullet}

\newcommand{\R}{{\mathbb{R}}}

\newcommand{\esp}{{\mathbb{E}}}

\newcommand{\proba}{{\mathbb{P}}}

\newcommand{\Rtensor}[2]{(\R^{#1})^{\otimes {#2}}}
\newcommand{\ind}{{\mathbf{1}}}

\newcommand{\pushright}[1]{\ifmeasuring@#1\else\omit\hfill$\displaystyle#1$\fi\ignorespaces}

\DeclareMathOperator*{\argmin}{arg\,min}

\DeclarePairedDelimiter\floor{\lfloor}{\rfloor}

\newcommand\interint[2]{\left\llbracket 1, 5\right\rrbracket}

\mathtoolsset{showonlyrefs}

\usepackage{tikz}
\usepackage{pgfplots}
\usepgfplotslibrary{fillbetween}
\usetikzlibrary{patterns}
\tikzset{>=latex}
\pgfplotsset{compat=1.16}
\usepackage{tkz-euclide}
\usepackage{xcolor}
\usetikzlibrary{calligraphy, calc, fit}
\usepackage{ifthen}
\usepackage{tabularx}
\usetikzlibrary{calc,positioning,matrix, decorations.pathreplacing}

\pgfdeclarelayer{bg}
\pgfsetlayers{bg,main}

\colorlet{lightgray}{gray!10}
\definecolor{colorp1}{HTML}{003f5c}
\definecolor{colorp2}{HTML}{58508d}
\definecolor{colorp3}{HTML}{ff6361}
\definecolor{colord1}{HTML}{66C2A5}
\definecolor{colord2}{HTML}{8DA0CB}
\definecolor{colord3}{HTML}{FC8D62}

\definecolor{colorp1}{HTML}{4878d0}
\definecolor{colorp2}{HTML}{ee854a}
\definecolor{colorp3}{HTML}{6acc64}
\definecolor{colord1}{HTML}{003f5c}
\definecolor{colord2}{HTML}{bc5090}
\definecolor{colord3}{HTML}{ff6361}

\tikzset{
    ncbar angle/.initial=90,
    ncbar/.style={
        to path=(\tikztostart)
        -- ($(\tikztostart)!#1!\pgfkeysvalueof{/tikz/ncbar angle}:(\tikztotarget)$)
        -- ($(\tikztotarget)!($(\tikztostart)!#1!\pgfkeysvalueof{/tikz/ncbar angle}:(\tikztotarget)$)!\pgfkeysvalueof{/tikz/ncbar angle}:(\tikztostart)$)
        -- (\tikztotarget)
    },
    ncbar/.default=0.5cm,
}

\tikzset{square left brace/.style={ncbar=0.5cm}}
\tikzset{square right brace/.style={ncbar=-0.5cm}}

\begin{document}

\title*{The insertion method to invert the signature of a path}
\author{Adeline Fermanian, Jiawei Chang, Terry Lyons and Gérard Biau}
\institute{Adeline Fermanian \at LOPF, Califrais’~Machine~Learning~Lab, Paris, France, \email{adeline.fermanian@califrais.fr}
\and Jiawei Chang \at Department of Mathematics, University of Oxford, United Kingdom, \email{jiawei.chang@hotmail.com}
\and Terry Lyons \at Department of Mathematics, University of Oxford, The Alan Turing Institute, United Kingdom, \email{terry.lyons@maths.ox.ac.uk}
\and G\'erard Biau \at Sorbonne Universit\'e, CNRS, Laboratoire de Probabilités, Statistique et Modélisation, F-75005 Paris, France, \email{gerard.biau@sorbonne-universite.fr}}
%
%

\maketitle

\abstract{The signature is a representation of a path as an infinite sequence of its iterated integrals. Under certain assumptions, the signature characterizes the path, up to translation and reparameterization. Therefore, a crucial question of interest is the development of efficient algorithms to invert the signature, i.e., to reconstruct the path from the information of its (truncated) signature. In this article, we study the insertion procedure, originally introduced by \citet{chang2019insertion}, from both a theoretical and a practical point of view. After describing  our version of the method, we give its rate of convergence for piecewise linear paths, accompanied by an implementation in Pytorch. The algorithm is parallelized, meaning that it is very efficient at inverting a batch of signatures simultaneously. Its performance is illustrated with both real-world and simulated examples.}

\section{Introduction}

To any multivariate path, that is, a smooth function $X:[0,1] \to \R^d$, $d \geqslant 1$, one can associate its signature, denoted by $S(X)$, which is a representation of $X$ as an infinite sequence of tensors of iterated integrals. The signature was first introduced by \citet{chen1957integration}, then was at the heart of rough path theory \citep{lyons2007differential,friz2010multidimensional}, and is now widely used in machine learning \citep{levin2013learning,primer2016,kidger2019deep,kiraly2016kernels,liao2019learning,toth2020bayesian,fermanian2021embedding}. The combination of signatures and machine learning algorithms has been successfully applied in various domains such as character recognition \citep{yang2016deepwriterid,wilson2018path}, finance \citep{lyons2014feature,perez2018derivatives}, human action recognition \citep{li2017lpsnet,yang2017leveraging}, medicine \citep{morrill2019sepsis,morrill2020utilization}, and emotion recognition \citep{wang2019path}.

The signature is a faithful representation of a path as a geometric object, in the sense that if $X$ has at least one monotone coordinate, then $S(X)$ uniquely determines $X$, up to translation and reparameterization \citep{hambly2010uniqueness}. A natural question, therefore, is to design methods and fast algorithms that can invert the signature, i.e., reconstruct $X$ from the information in $S(X)$. This question is of particular interest in the area of time series generation, where time series over a fixed horizon can be seen as paths. The goal is to generate (fake) time series that follow the same distribution as those of a dataset. Similar to the moment-generating function for random variables, the expected signature characterizes distributions of paths \citep{chevyrev2018signature}, making it a very attractive tool for generative modeling \citep{ni2020conditional,lou2023pcf}. Thus, instead of generating time series directly, a possible strategy is to generate signatures and then invert them to get back to the time domain \citep{buehler2020data}. Moreover, beyond the theoretical interest, we can also think of several other applications of signature inversion such as trend estimation or data smoothing.

Inverting signatures has been an active area of research in recent years, with methods often being theoretically oriented, but practically intractable or slow in their algorithmic execution. For example, \citet{lyons2017hyperbolic} provide a procedure for piecewise linear paths, \citet{lyons2018inverting} for ${\mathscr C}_1$ paths, and \citet{chang2017signature} for monotone paths. Despite their mathematical interest, none of these methods comes with a package that allows them to be implemented and used in a data processing pipeline. For completeness, we also mention that some authors \citep[e.g., ][]{kidger2019deep, buehler2020data} have proposed neural-network-based inversion strategies, but without any theoretical guarantees and at a high computational cost.

In this paper, inspired by the preliminary results of \citet{chang2019insertion}, we propose a new, fast, and computationally tractable method, which we call the insertion method. This method exploits the property that signatures of order $n$ and $n+1$ are connected by inserting the derivatives of the path into a tensor product, which is then integrated. Retrieving the derivatives of a path can then be written as a linear minimization problem, which can be solved very efficiently. We show through a sound theoretical analysis that this approach allows to reconstruct any piecewise linear path with a given rate of convergence, asymptotically in $n$. Working with piecewise linear paths is by no means a restriction since in practice the signature is always computed by linear interpolation of the measurement points of $X$. Besides the theoretical guarantees, our method is accompanied by a fast algorithm that, given the signature of a piecewise linear path of order $n$, outputs a piecewise linear path of $n$ pieces. This algorithm is implemented in the PyTorch framework \citep{paszke2019pytorch}, and is now part of the Signatory \citep{signatory} package, freely available to the statistics and machine learning community.

The rest of the paper is organized as follows. In Section \ref{sec:preliminaries} we recall some basic definitions and elementary results on signatures and tensors, which are essential for a good understanding of our approach. Section \ref{sec:insertion_algorithm} describes the insertion method, gives the main theoretical results, and presents the algorithm. Section \ref{sec:experimental_results} shows our method in action through some examples of signature inversion, together with a brief study of the computational cost of the algorithm. For the sake of clarity, most of the technical proofs have been moved to the Supplementary Material.

\section{Preliminaries} \label{sec:preliminaries}

In this section, we introduce the mathematical framework along with some of the notation used throughout the paper. In all the following, we assume that $\R^d$ is equipped with the Euclidean norm, denoted by $\| \cdot \|$. We refer to \citet[][Chapter 1]{friz2010multidimensional} for further details on paths of bounded variation, and to \citet{lyons2007differential} for an introduction to signatures and tensors.

\subsection{Paths of bounded variation and admissible norms} 
\label{subsec:bv_paths}

\begin{definition}
	\label{def:bv}
	Let $X:[0,1] \rightarrow \R^d$ be a continuous path. For any $[u,v] \subset [0,1]$, the total variation (or length) of $X$ on $[u,v]$ is defined by
	\[\|X\|_{TV; [u,v]} = \underset{(u_0, \dots, u_n) \in D_{u,v}}{\textnormal{sup}} \  \sum_{i=1}^n \|X_{u_i} - X_{u_{i-1}} \| ,\]
	where $D_{u,v}$ denotes the set of all finite partitions of $[u,v]$, i.e.,
	\[D_{u,v}=\big\{ (u_0,\ldots,u_n) \mid n \geqslant 1,\, u=u_0 \leqslant u_1\leqslant \cdots \leqslant u_{n-1} \leqslant u_n=v \big\}.\]
	The path $X$ is said to be of bounded variation on $[u,v]$ if its total variation is finite.
\end{definition}

Throughout the document, we denote by $BV(\R^d)$ the set of continuous paths of bounded variation on $[0,1]$ with values in $\R^d$. Moreover, when $[u,v]=[0,1]$, we often write $\|X\|_{TV}$ instead of $\|X\|_{TV; [0,1]}$. 

\begin{example}{Example}
Let $X:[0,1] \to \R^d$ be a continuous piecewise linear path, and let $0=t_0< t_1< \dots < t_{M-1}<t_M=1$
be the minimal partition such that $X$ is linear on each $[t_{i-1},t_i]$. So, there exist $\alpha_1, \dots, \alpha_M, \beta_1, \dots, \beta_M \in \R^d$ such that
\begin{equation*}
X_t = \alpha_i +\beta_i t, \quad \text{ for } t \in [t_{i-1},t_i], \quad i \in \{1, \dots, M \}.
\end{equation*}
One has
\begin{equation*}
\|X\|_{TV} = \sum_{i=1}^M \|\beta_i\| (t_{i-1}-t_i).
\end{equation*}
\end{example}

For $E$ and $F$ two finite-dimensional real vector spaces, their tensor product is denoted by $E \otimes F$ \citep[see][]{purbhoo2012notes}. It is also a vector space, and, if $(e_i)_{i \in I}$, $(f_j)_{j \in J}$, are bases of $E$ and $F$ respectively, then $(e_i \otimes f_j)_{i \in I, j \in J}$ is a basis of $E \otimes F$. If $\textnormal{dim}(E) = d$, $\textnormal{dim}(F) = e$, then $\textnormal{dim}(E \otimes F) = d \times e$.

For $n \geqslant 0$, the $n$th tensor power of $E$ is defined as the order $n$ tensor product of $E$ with itself, i.e., $E^{\otimes n}= E \otimes \cdots \otimes E$, with the convention $E^{\otimes 0} = \R$. If $(e_1,\ldots,e_d)$ is a basis of $E$, then any element $a$ of $E^{\otimes n}$ can be written as 
\begin{equation*}
a = \sum_{(i_1,\dots,i_n) \in \{1,\dots,d\}^n} a_{(i_1,\dots,i_n)} e_{i_1} \otimes \cdots \otimes e_{i_n}, \qquad a_{(i_1,\dots,i_n)} \in \R.
\end{equation*}
The space $E^{\otimes n}$ is of dimension $d^n$, which means that we can identify $E^{\otimes n}$ with $\R^{d^n}$. In particular, $E^{\otimes 2}$ can be identified with the space of $(d\times d)$ matrices.

From now on, it is assumed that $E = \R^d$, where $\R^d$ is equipped with its canonical basis, denoted in this article by $(e_1, \dots, e_d)$. We want to equip the tensor spaces $(\R^d)^{\otimes n}$, $n \geqslant 1$, with norms inherited from the Euclidean norm on $\R^d$. An essential condition is that these norms should behave ``well'' with respect to the tensor product, a property summarized by the notion of admissible norms. This important notion is based on the notion of permutation on $\Rtensor{d}{n}$, recalled in the definition below. 
\begin{definition}
\label{def:tensor_permutation}
Let $\sigma$ be a permutation of $\{ 1, \dots, n \}$, i.e., a bijective function of $\{1, \dots, n \}$ over itself. For any $a \in \Rtensor{d}{n}$ of the form
\[a = \sum_{(i_1, \dots, i_n) \in \{1, \dots, d\}^n} a_{(i_1, \dots, i_n)} e_{i_1}\otimes \dots \otimes e_{i_n}, \] 
we let
\begin{equation*}
	\sigma(a)=  \sum_{(i_1, \dots, i_n) \in \{ 1, \dots, d\}^n} a_{(i_1, \dots, i_n)} e_{i_{\sigma(1)}}\otimes \dots \otimes e_{i_{\sigma(n)}}.
\end{equation*}
\end{definition}
\begin{example}{Example}
Let $n=2$, $\sigma(1)=2$, and  $\sigma(2)=1$. Then, for, any $a \in (\R^2)^{\otimes 2}$, one has
\begin{align*}
a&=a_{(1,1)} e_1 \otimes e_1 + a_{(1,2)} e_1 \otimes e_2 + a_{(2,1)} e_2 \otimes e_1  + a_{(2,2)} e_2 \otimes e_2,
\end{align*}
and thus
\begin{align*}
\sigma(a)&=a_{(1,1)} e_2 \otimes e_2 + a_{(1,2)} e_2 \otimes e_1 + a_{(2,1)} e_1 \otimes e_2  + a_{(2,2)} e_1 \otimes e_1.
\end{align*}
\end{example}

\begin{definition}[Admissible norms]
	\label{def:admissible_norm}
	Assume that, for each $n \geqslant 1$, $(\R^d)^{\otimes n}$ is endowed with a norm $ \| \cdot \|_n$. The norms $(\| \cdot \|_n)_{n \geqslant 1}$ are said to be admissible if
	\begin{enumerate}
		\item[$(i)$] For any $n \geqslant 1$, any permutation $\sigma$ of $\{1, \dots, n \}$,  any $a \in (\R^d)^{\otimes n}$,
		\[\|\sigma(a) \|_n = \| a\|_n. \]
		\item[$(ii)$] For any $n,m \geqslant 1$, any $a \in (\R^d)^{\otimes n}$, any $b \in (\R^d)^{\otimes m}$,
		\begin{equation*}
		\| a \otimes b\|_{n + m} = \| a\|_n \|b\|_m .
		\end{equation*}
	\end{enumerate}
\end{definition}

Throughout the paper, to keep the vocabulary simple, we refer to ``an admissible norm'' instead of ``a collection of admissible tensor norms'' and, since no confusion is possible, we write $\| \cdot \|$ instead of $\| \cdot \|_n$. There are several admissible norms. The most natural one is the Euclidean tensor norm, defined for $a \in \Rtensor{d}{n}$ by
	\begin{equation*}
    \|a\| = \Big(\sum_{(i_1, \hdots, i_n) \in \{1,\dots,d\}^n} a_{(i_1, \hdots, i_n)}^2 \Big)^{\nicefrac{1}{2}}.
	\end{equation*}
Another important admissible norm is the projective norm, which is the largest possible admissible norm. Denote by $\| \cdot \|$ any admissible norm, and write any element $a \in (\R^d)^{\otimes n}$ as 
\begin{equation} \label{eq:representation_tensor_element}
a = \sum_{i=1}^k a_{1,i} \otimes \dots \otimes a_{n,i}, \qquad a_{1,i}, \dots, a_{n,i} \in \R^d.
\end{equation}
Note that such a representation always exists but is clearly not unique. Then, for statement $(ii)$ in Definition \ref{def:admissible_norm} to be satisfied, by the triangle inequality, we must have
\[ \|a \| \leqslant \sum_{i=1}^k \|a_{1,i}\| \cdots \| a_{n,i} \|. \]
Since this inequality is valid for any representation of the form \eqref{eq:representation_tensor_element}, we must therefore have
\[ \|a \| \leqslant \inf \Big\{ \sum_{i=1}^k \|a_{1,i}\|\cdots \| a_{n,i} \| \mid a=\sum_{i=1}^k a_{1,i} \otimes \dots \otimes a_{n,i}, \, k \geqslant 1 \Big\}.\]
It is easily shown that
\begin{equation*}\label{eq:def_projective_norm}
\| a \|_{\pi}= \inf \Big\{ \sum_{i=1}^k \|a_{1,i}\|\cdots \| a_{n,i} \| \mid a=\sum_{i=1}^k a_{1,i} \otimes \dots \otimes a_{n,i}, \, k \geqslant 1\Big\}
\end{equation*}
is an admissible norm, called the projective norm. We refer to \citet{ryan2013introduction} for more details on tensor norms.

Throughout the article, it will be assumed that the tensor powers of $\R^d$ have been equipped with an admissible norm. We close the subsection with the definition of the tensor algebra.
\begin{definition}
	The tensor algebra $T(\R^d)$ is defined by
	\begin{equation*}
	 T(\R^d) = \big\{ (a_0,\dots,a_k,\dots) \mid \forall k \geqslant 0, a_k \in (\R^d)^{\otimes k} \big\}.
	 \end{equation*}
	For $n\geqslant 0$, the truncated tensor algebra is
	\begin{equation*}
	T^n(\R^d)= \big\{ (a_0,\dots,a_n) \mid \forall k \in \{0, \dots, n\}, a_k \in (\R^d)^{\otimes k} \big\}.
	\end{equation*}
\end{definition}

\subsection{The signature of a path}
\label{sec:sig_def_notations}
We are now in a position to define the signature of a path of bounded variation and to give some elementary examples.
\begin{definition}[Signature of a path]
\label{def:signature}
	Let $X \in BV(\R^d)$. The signature of $X$ is defined by
	\begin{equation*}
		S(X) = (1,{\mathbf{X}^1}, {\mathbf{X}^2},\dots,{\mathbf{X}^n},\dots ) \in T(\R^d) ,
	\end{equation*}
	where, for each $n \geqslant 1$,
	\begin{equation*}
		{\mathbf{X}^n} = \idotsint\limits_{ 0 \leqslant u_1 \leqslant \cdots \leqslant u_n \leqslant 1 } dX_{u_1}\otimes \dots \otimes dX_{u_n} \in (\R^d)^{\otimes n}.
	\end{equation*}
\end{definition}

Recall that when $X$ is a continuously differentiable path, then ``$dX_t =X'_t dt$'', where $X'$ is the derivative of $X$ with respect to $t$. However, the integrals in Definition \ref{def:signature} are still well defined for continuous paths of bounded variation using the notion of Riemann-Stieljes integrals \citep[e.g.,][]{lyons2007differential}.
\begin{example}{Example}\renewcommand{\arraystretch}{1.5}
	Assume that $d=2$ and let $X_t=(X^1_t, X^2_t)$. Since $(\R^2)^{\otimes n}$ can be identified with $\R^{2^n}$, the signatures of order 1 and 2 are equal to
	\begin{equation*}
	{\mathbf{X}^1} =\int_{0}^{1}dX_t = \begin{pmatrix}
	\int_{0}^{1} dX^{1}_t \\
	\int_{0}^{1} dX^{2}_t \\
	\end{pmatrix}
	\end{equation*}
	\begin{equation*}
	{\mathbf{X}^2}=\int_{0}^{1} \int_{0}^{t} dX_s \otimes dX_t  = \begin{pmatrix}
	\int_{0}^{1} \int_{0}^{t} dX^1_s dX^1_t & \int_{0}^{1} \int_{0}^{t} dX^1_s dX^2_t \\
	\int_{0}^{1} \int_{0}^{t} dX^2_s dX^1_t & \int_{0}^{1} \int_{0}^{t} dX^2_s dX^2_t \\
	\end{pmatrix}.
	\end{equation*}
\end{example}
For any $n \geqslant 1$, the truncated signature of order $n$ is denoted by
	\begin{equation*}
	S^n(X) =(1,{\mathbf{X}^1}, \dots,{\mathbf{X}^n}).
	\end{equation*}
For a multi-index $I = (i_1,\dots,i_n) \in \{1,\dots,d\}^n$, the coefficient of $\mathbf{X}^n$ corresponding to $I$ is
	\begin{equation*}
	S^I(X) =  \int_{(u_1, \dots, u_n) \in \Delta_n } dX^{i_1}_{u_1} \dots dX^{i_n}_{u_n},
	\end{equation*}
	where $X_t = (X^1_t, \dots, X^d_t)$, $t \in [0,1]$, and
	\begin{equation}
	\label{eq:def_Delta_k}
	\Delta_{n} = \big\{(u_1,\dots,u_n) \in [0,1]^n \mid 0 \leqslant u_1 \leqslant \cdots \leqslant u_n \leqslant 1\big\}.
	\end{equation}
	With this notation, the signature of order $n$ takes the form
	\[\mathbf{X}^n = \sum_{I \in \{1, \dots, d \}^n} S^I(X) e_{i_1} \otimes \dots \otimes e_{i_n}.\]
We will sometimes consider the signature of a path restricted to a certain interval $[u,v] \subset [0,1]$. Its signature is then denoted by $S_{[u,v]}(X)$ (respectively $\mathbf{X}^n_{[u,v]}$, and so on). Similarly, the simplex in $[u,v]^n$ is denoted by
	\begin{equation*}
	\label{eq:def_Delta_k_u_v}
	\Delta_{n; [u,v]} = \big\{(u_1,\dots,u_n) \in [u,v]^n \mid u \leqslant u_1 \leqslant \cdots \leqslant u_n \leqslant v\big\}.
	\end{equation*}
Note that the signature $S(X)$ is an element of $T(\R^d)$ and that the truncated signature $S^n(X)$ is an element of $T^n(\R^d)$. 
\begin{example}{Example}
	Let $X:[0,1] \to \R^d$ be a linear path, i.e., $X_t = \alpha + \beta t$, $\alpha, \beta \in \R^d$. Then, for any $n \geqslant 1$, $[u,v] \subset [0,1]$,
	\begin{align}\label{eq:signature_linear_path}
	\mathbf{X}^n_{[u,v]} &= \int_{\Delta_{n; [u,v]}}dX_{u_1} \otimes \dots \otimes  dX_{u_n} \nonumber \\
	&=  \int_{\Delta_{n; [u,v]}} (\beta du_1) \otimes \dots \otimes (\beta du_n)  \nonumber \\
	&= \beta^{\otimes n} \int _{\Delta_{n; [u,v]}}du_1 \dots du_n = \beta^{\otimes n} \frac{ (v-u)^n}{n!}. 
	\end{align}
\end{example}

We now recall without proofs the properties of signatures that will be essential to our method of inversion, and refer to \citet[][Chapter 2]{lyons2007differential} for more material.

\begin{proposition} \label{prop:basic_sig_prop}
$ $
\begin{itemize}
	\item[$(i)$] Let $X \in BV(\R^d)$, $u,v,w \in [0,1]$ such that $u \leqslant v \leqslant w$. Then
	\begin{equation*}
	\mathbf{X}^n_{[u,w]} = \sum_{k=0}^n \mathbf{X}^k_{[u,v]} \otimes \mathbf{X}^{n-k}_{[v,w]}.
	\end{equation*}
	\item[$(ii)$] Let $X \in BV(\R^d)$, $\psi : [0,1] \rightarrow [0,1]$ a reparameterization (i.e., a continuous non-decreasing surjection), and $\widetilde{X}_t = X_{\psi(t)} $ for any $t \in [0,1]$. Then $S(\widetilde{X}) = S(X).$
\end{itemize}
\end{proposition}

The first statement of Proposition \ref{prop:basic_sig_prop} is known as Chen's identity \citep[][Theorem 2.9]{lyons2007differential}. It provides a procedure for computing signatures of piecewise linear paths. Indeed, starting from the explicit formula for the signature of a linear path given by \eqref{eq:signature_linear_path}, this identity allows to compute the signature of a concatenation of two linear sections. Iterating this scheme allows to compute the signature of any piecewise linear path. This is the basis of the tools available in Python such as esig, iisignature \citep{reizenstein2018iisignature}, and Signatory \citep{signatory}.

Since we are interested in reconstructing a path from its signature, we need to discuss what it means for two paths to have the same signature. It is clear from their definition that signatures are invariant by translation, and according to Proposition \ref{prop:basic_sig_prop}, $(ii)$ they are also invariant by reparameterization. The appropriate notion to encapsulate how signatures characterize paths is the so-called property of tree-like equivalence, introduced by \citet{hambly2010uniqueness} and defined as follows. The notation $\overleftarrow{X}$ denotes the time-reversal of a path $X$, defined, for any $t \in [0,1]$, by $\overleftarrow{X}_t = X_{1-t}$.

\begin{definition}[Tree-like equivalence]
$ $
	\label{intro:def:tree_like}
	\begin{itemize}
	\item[$(i)$] A path $X\in BV(\R^d)$ is tree-like if there exists a continuous function $h:[0,1] \to [0,+\infty)$ such that $h(0) = h(1)=0$ and, for any $[s,t] \subset [0,1]$,
	\begin{equation*}
	\|X_s - X_t \| \leqslant h(s) + h(t) -2 \underset{u \in [s,t]}{\textnormal{inf}} h(u).
	\end{equation*}
	\item[$(ii)$] Two paths $X,Y \in BV(\R^d)$ are tree-like equivalent if $X*\overleftarrow{Y}$ is a tree-like path. This relation is denoted by $X \sim Y$.
	\end{itemize}
\end{definition}

Informally, a tree-like path is a path that retraces itself backward. This concept is important insofar as it is involved in the following uniqueness theorem, due to \citet[][Theorem 4, Corollary 1.6]{hambly2010uniqueness}.

\begin{theorem}[Uniqueness]
	\label{th:uniqueness}
	Let $X, Y \in BV(\R^d)$. Then $S(X)=S(Y)$ if and only if $X \sim Y$. Moreover, for any $X \in BV(\R^d)$, there exists a unique path (up to translation and reparameterization) of minimal length in its equivalence class, denoted by $\overline{X}$ and called the reduced path.
\end{theorem}

Theorem \ref{th:uniqueness} can be summarized as follows: the signature loses the information about the initial position of the path, of the time parameterization, and of any tree-like piece. We conclude this subsection with Proposition \ref{prop:up_bound_norm_sig}, which gives an upper bound on the size of signature coefficients.

\begin{proposition}
	\label{prop:up_bound_norm_sig}
	Let $X \in BV(\R^d)$, $[u,v] \subset [0,1]$. Then, for any admissible tensor norm and $n \geqslant 1$, one has
	\[\|{\mathbf{X}^n_{[u,v]}} \| \leqslant \frac{ \| X\|_{TV; [u,v]}^n}{n!}.\]
\end{proposition}

\section{The insertion method} \label{sec:insertion_algorithm}

\subsection{Description}

We are now ready to describe the insertion method and its associated algorithm. Suppose that we are given a piecewise linear path $X:[0,1] \to \R^d$ as in Subsection \ref{subsec:bv_paths}. Recall that we denote by 
\[0=t_0< t_1< \dots < t_{M-1}<t_M=1\]
the minimal partition such that $X$ is linear on each $[t_{i-1},t_i]$. Thus, there exist intercepts $\alpha_1, \dots, \alpha_M \in \R^d$ and non-zero slopes $\beta_1, \dots, \beta_M \in \R^d$ such that 
\begin{equation} \label{eq:def_X_piecewise_linear}
X_t = \alpha_i +\beta_i t, \quad \text{ for } t \in [t_{i-1},t_i], \quad i \in \{1, \dots, M \}.
\end{equation}
For each $i \in \{1, \dots, M-1\}$, we let $\omega_i$ be the angle between the linear segments $[t_{i-1}, t_i]$ and $[t_i, t_{i+1}]$, i.e.,
\begin{equation*}
	\omega_i = \text{Arccos}\Big( \frac{\langle \beta_{i}, \beta_{i+1} \rangle}{\|\beta_i\| \|\beta_{i+1}\|} \Big) \in [0, \pi],
\end{equation*}
where $\langle \cdot, \cdot \rangle$ is the Euclidean scalar product and $\text{Arccos}$ is the arccosine function. Assuming that the partition is minimal means that $\omega_i \neq 0$. In addition, we assume that $X$ is the reduced path in its equivalence class, for the tree-like equivalence (see Theorem \ref{th:uniqueness}),  which means that $\omega_i \neq \pi$. (If $\omega_i=\pi$, then segment $i+1$ retraces backwards segment $i$ and is therefore a tree-like part of the path, which is excluded if $X$ is reduced.) In summary, we have the following assumption:
\begin{equation*}\label{eq:def_omega}
	(A_1) \quad  \forall i \in \{1, \dots, M\}, \quad \omega_i \in \, ]0, \pi[.
\end{equation*}

We denote by $\omega= \min_i \omega_i$ the smallest angle between two linear segments. Invariance by translation means that we cannot recover $\alpha_0$ from $S(X)$. Due to the continuity of $X$, the $\alpha_i$ for $i \in \{1, \dots, M\}$ are completely determined by $\alpha_0, \beta_1, \dots, \beta_M$. Therefore, the only information we hope to recover is the information on the slopes $\beta_i$, up to the reparameterization of the path. So we choose a specific parameterization of $X$. Let $\ell = \|X\|_{TV}$ be the length of $X$, assumed to be strictly positive, and let $\phi$ be the reparameterization defined as
\begin{align*}
	\phi: [0,1] &\to [0,1] \\
			t & \mapsto \frac{\| X \|_{TV; [0,t]}}{\ell}.
\end{align*}
The function $\phi$ is strictly increasing and continuous piecewise linear. In fact, if $t \in [t_{i-1}, t_i]$, then 
\[\phi(t) = \ell^{-1}\sum_{k=1}^{i-1} \| \beta_k\| (t_k - t_{k-1}) + \ell^{-1}\|\beta_i\| (t-t_{i-1}).\] 
The path $\overline{X}_t = X_{\phi^{-1}(t)}$ is therefore a reparameterization of $X$, and $S(X) = S(\overline{X})$. More precisely, $\overline{X}$ is piecewise linear on the partition $0=u_0 < u_1 < \dots < u_M = 1$, where
\begin{equation*}
	u_i = \ell^{-1}\sum_{k=1}^{i} \|\beta_k\| (t_k - t_{k-1}).
\end{equation*}
To see this, just note that for each $i \in \{1,\dots, M\}$ and $u \in [u_{i-1}, u_i]$, one has
\[\phi^{-1}(u) = \frac{\ell}{\| \beta_i\|} (u-u_{i-1}) + t_{i-1},\]
and
\[ \overline{X}_u = \alpha_i + \beta_i \phi^{-1}(u) = \alpha_i + \frac{\ell}{\| \beta_i\|} \beta_i (u-u_{i-1}) + \beta_i t_{i-1}.\]
The path $\overline{X}$ is thus piecewise linear and its slopes have a constant norm, equal to $\ell$. Note that recovering $X$ up to translation and reparameterizations amounts to recovering the direction of each linear segment, i.e., $\nicefrac{\beta_1}{\| \beta_1\|}, \dots, \nicefrac{\beta_M}{\| \beta_M\|}$. So, from now on, we will assume that $X$ is piecewise linear on a partition $t_0 = 0 < t_1 < \dots < t_M = 1$ and that $\|\beta_i\| = \ell$ for each $i \in \{1, \dots, M\}$:
\begin{equation*}
	(A_2) \quad \forall i \in \{1, \dots, M\}, \quad \|\beta_i\| = \ell.
\end{equation*}

Throughout, to avoid degenerate situations, it is assumed that for each $n \geqslant 1$, $\| \mathbf{X}^n \| \neq 0$. There is no harm in doing so since according to \citet{boedihardjo2019non} this assumption is satisfied for $n$ large enough when $X$ is not tree-like, which is our case under Assumption $(A_1)$. The general idea of the insertion method is as follows. Signatures of order $n$ and $n+1$ are connected by inserting the derivatives of the path into a tensor product, which is then integrated. Therefore, retrieving the derivative of the path boils down to finding a vector in $\R^d$ such that, once inserted into the tensor product of the signature of order $n$, minimizes the distance between that modified signature and the signature of order $n+1$. This can be written as a linear minimization problem, which can be solved very efficiently. It allows to retrieve the slopes $\beta_i$ of the path at different intervals and thus to reconstruct the path by integration.

The first ingredient of the method is the insertion map. It is a linear map defined as a tensor product between a vector $y \in \R^d$ and the signature of $X$ of order $n$ along the dimension $p$.

\begin{definition}[Insertion map]
\label{def:insertion_map}
Let $X \in BV(\R^d)$, $n \geqslant 1$, and $p \in \{1, \dots, n+1\}$. The insertion map $\mathscr{L}^n_{p,X}:\R^d \rightarrow (\R^d)^{\otimes (n+1)}$ is defined, for any $y \in \R^d$, by
\begin{align*}
&[p=1]\quad {\mathscr L}^n_{1,X}(y) = \int_{ \Delta_n} y \otimes dX_{u_1} \otimes \dots \otimes dX_{u_n},  \\
& [2 \leqslant p <n]\quad  {\mathscr L}^n_{p,X}(y)= \int_{\Delta_n} dX_{u_1} \otimes \dots \otimes dX_{u_{p-1}} \otimes y \otimes dX_{u_{p}} \otimes \dots \otimes dX_{u_n}, \\
& [p=n+1] \quad {\mathscr L}^n_{n+1,X}(y)= \int_{\Delta_n} dX_{u_1} \otimes \dots \otimes dX_{u_n} \otimes y.
\end{align*}
\end{definition}

A first observation is that the map ${\mathscr L}^n_{p,X}$ is linear. Moreover, it depends only on the signature of $X$ of order $n$ and is Lipschitz continuous for any admissible tensor norm (see Definition \ref{def:admissible_norm}). These properties are summarized in the next proposition.
\begin{proposition} 
\label{insertiontest}
Let $X\in BV(\R^d)$. For any $n \geqslant 1, p \in \{1, \dots, n+1\}$, $\mathscr{L}^n_{p,X}$ is linear, Lipschitz-continuous, and its Lipschitz constant is $\| \mathbf{X}^n \|$.
\end{proposition}
\begin{proof}
The linearity is an immediate consequence of the linearity of the tensor product. Moreover, for $y,z \in \R^d$, we have
\begin{align*}
&\|\mathscr{L}^n_{p,X}(y) - \mathscr{L}^n_{p,X}(z)\| \nonumber\\
&\quad =\big\| \int_{\Delta_n} dX_{u_1} \otimes \dots \otimes dX_{u_{p-1}} \otimes (y-z) \otimes dX_{u_{p}}\otimes \dots \otimes dX_{u_n} \big \| \nonumber \\
 &\quad =\big\| \int_{\Delta_n} dX_{u_1} \otimes \dots \otimes dX_{u_n} \otimes (y-z) \big \| \tag*{(by Definition \ref{def:admissible_norm}, $(i)$)} \nonumber \\
& \quad = \big\| \int_{\Delta_n} dX_{u_1} \otimes \dots \otimes dX_{u_n} \big \| \cdot  \| y-z \| \tag*{(by Definition \ref{def:admissible_norm}, $(ii)$)} \nonumber \\
& \quad = \| \mathbf{X}^n \| \cdot \| y-z\|. 
\end{align*}
{\qed}
\end{proof}

The general idea of the insertion method is to solve the following optimization problem: 
\begin{align} \label{eq:min_problem_insertion_algo}
    y^\ast_{p,n} \in \argmin_{y \in \R^d} \|\mathscr{L}^n_{p,X}(y)- (n+1)\mathbf{X}^{n+1}\|.
\end{align}
The main result of the section is that the minimizer of \eqref{eq:min_problem_insertion_algo} is close to $\beta_i$ for the projective norm, provided the parameter $\nicefrac{p}{(n+1)}$ is well-chosen inside $]t_{i-1}, t_i[$. 
\begin{theorem}
\label{th:insertion_cvg}
	Let $X$ be a piecewise linear path as defined by \eqref{eq:def_X_piecewise_linear}, following assumptions $(A_1)$ and $(A_2)$. For any $n\geqslant1$, $p \in \{1, \dots, n+1\}$, define
	\begin{equation*}
		y^\ast_{p,n} \in \argmin_{y \in \R^d} \|\mathscr{L}^{n}_{p,X}(y)- (n+1)\mathbf{X}^{n+1}\|_\pi.
	\end{equation*}
	Let
	\begin{equation*}\label{eq:def_K_omega_and_n_1}
		K(\omega) = \log\Big(\frac{2}{1 - \cos(\nicefrac{\omega}{2})} \Big) \quad \text{ and } \quad n_1 = \lfloor4 e^{2(M-1) K(\omega)}\rfloor.
	\end{equation*}
Then, for any $ i \in \{1, \dots, M \}$, $n > \max (n_1, \nicefrac{2}{(t_i - t_{i-1})})$, and $p=\lfloor \nicefrac{(3t_i + t_{i-1})(n+1)}{4} \rfloor$, there exists an integer $k_n \in [n-n^{\nicefrac{3}{4}}, n + n^{\nicefrac{3}{4}}]$ such that
	\begin{equation*}
		\| y^{\ast}_{p,k_n}- \beta_i \| \leqslant 4 \ell e^{(M-1)K(\omega)}\Big(\frac{1}{\sqrt{k_n+1}} \sqrt{\frac{1 - (t_i - t_{i-1})}{t_i - t_{i-1}}} +  4\exp\Big(- \frac{k_n}{16}(t_i-t_{i-1})^2\Big) \Big).
	\end{equation*}
\end{theorem}

Since the intervals $[n-n^{\nicefrac{3}{4}}, n + n^{\nicefrac{3}{4}}]$ shift towards infinity with a width that grows slower than $n$, we have the following corollary.

\begin{corollary}
	For any $i \in \{1, \dots, M\}$, there exists a strictly increasing sequence of integers $(k_n)$ and a sequence of positive integers $(p_n)$ such that $y^\ast_{p_n,k_n}$ converges to $\beta_i$ as $n$ increases.
\end{corollary}

Let us comment on Theorem \ref{th:insertion_cvg}. 
To simplify the interpretation, consider the case of a regular partition, where $t_i = \nicefrac{i}{M}$. Then, Theorem \ref{th:insertion_cvg} becomes: for any $n > \max(n_1, 2M)$ and $p = \nicefrac{(2M-1)(n+1)}{4M}$,
\begin{equation}
\label{eq:bound_thm_regular_case}
		\| y^{\ast}_{p,k_n}- \beta_i \| \leqslant 4 \ell e^{(M-1)K(\omega)}\Big(\sqrt{\frac{M-1}{k_n+1}} +  4\exp\Big(- \frac{k_n}{16M^2}\Big) \Big).
\end{equation}
As the number $M$ of piecewise linear pieces increases (or, equivalently, the length of the interval $t_i - t_{i-1}$ decreases), it becomes harder to recover $\beta_i$---both the right-hand-side of \eqref{eq:bound_thm_regular_case} and the truncation order $n$ necessary for the result to hold increase.

Theorem \ref{th:insertion_cvg} relies heavily on the following two propositions, the proofs of which are given in the Supplementary Material.
\begin{proposition}
\label{thm:insertion_X'_theta_cvg}
Let $X$ be a piecewise linear path as defined by \eqref{eq:def_X_piecewise_linear}, following assumptions $(A_1)$ and $(A_2)$. For any $ i \in \{1, \dots, M \}$, $n \geqslant \nicefrac{2}{(t_i - t_{i-1})}$, and $p=\lfloor \nicefrac{(3t_i + t_{i-1})(n+1)}{4} \rfloor$, one has
\begin{align*}
	&\| \mathscr{L}^n_{p,X}(\beta_i) - (n+1)\mathbf{X}^{n+1}\| \\
 & \quad \leqslant \frac{\ell^{n+1}}{n!} \Big(\frac{1}{\sqrt{n+1}} \sqrt{\frac{1 - (t_i - t_{i-1})}{t_i - t_{i-1}}} +  4\exp\Big(- \frac{n}{16}(t_i-t_{i-1})^2\Big) \Big).
\end{align*}
\end{proposition}

\begin{proposition}
\label{th:sig_lower_bound}
Let $X$ be a piecewise linear path as defined by \eqref{eq:def_X_piecewise_linear}, following assumptions $(A_1)$ and $(A_2)$. For any $n > n_1$, there exists an integer $k_n \in [n-n^{\nicefrac{3}{4}}, n + n^{\nicefrac{3}{4}}]$ such that
\begin{equation*}
	\| \mathbf{X}^{k_n}\|_\pi \geqslant \ell^{k_n}\frac{e^{-(M-1)K(\omega)}}{2k_n!},
\end{equation*}
where $\| \cdot \|_\pi$ is the projective norm.
\end{proposition}

\begin{proof}[Theorem \ref{th:insertion_cvg}]
	The identity of Proposition \ref{insertiontest} is true for every admissible norm so it is especially true for the projective norm $\| \cdot \|_\pi$. Thus, for any $n \geqslant 1$,
	\begin{align*}
		\|\mathscr{L}^n_{p,X}(y^{\ast}_{p,n}) - \mathscr{L}^n_{p,X}(\beta_i) \|_\pi = \| \mathbf{X}^n\|_\pi \, \| y^{\ast}_{p,n}-\beta_i \|.
	\end{align*}
	Proposition \ref{th:sig_lower_bound} guarantees the existence of an integer $k_n \in [n-n^{\nicefrac{3}{4}},n+n^{\nicefrac{3}{4}}]$ such that
	\begin{equation*}
	\| \mathbf{X}^{k_n}\|_\pi \geqslant \ell^{k_n}\frac{e^{-(M-1)K(\omega)}}{2k_n!}.
\end{equation*}
	Moreover, Proposition \ref{thm:insertion_X'_theta_cvg} (which is valid for any tensor norm) yields 
	\begin{align*}
	     & \|\mathscr{L}^{k_n}_{p,X}(\beta_i)-(k_n+1)\ \mathbf{X}^{k_n+1} \|_\pi \\
      & \quad \leqslant \frac{\ell^{k_n+1}}{k_n!} \Big(\frac{1}{\sqrt{k_n+1}} \sqrt{\frac{1 - (t_i - t_{i-1})}{t_i - t_{i-1}}} +  4\exp\Big(- \frac{k_n}{16}(t_i-t_{i-1})^2\Big) \Big).
	\end{align*}
	By the very definition of $y^\ast_{p,n}$, we have
	\[ \|\mathscr{L}^{k_n}_{p,X}(y^\ast_{p,k_n})-(k_n+1)\ \mathbf{X}^{k_n+1} \|_\pi \leqslant  \|\mathscr{L}^{k_n}_{p,X}(\beta_i)-(k_n+1)\ \mathbf{X}^{k_n+1} \|_\pi.\]
Combining these inequalities, we obtain
	\begin{align*}
	 \| y^{\ast}_{p,k_n}- \beta_i \| & = \frac{1}{\| \mathbf{X}^{k_n}\|_\pi} \|\mathscr{L}^{k_n}_{p,X}(y^{\ast}_{p,k_n}) - \mathscr{L}^{k_n}_{p,X}(\beta_i) \|_\pi \\
	 & \leqslant  \frac{1}{\| \mathbf{X}^{k_n}\|_\pi} \big( \|\mathscr{L}^{k_n}_{p,X}(y^{\ast}_{p,k_n}) - (k_n+1)\ \mathbf{X}^{k_n+1} \|_\pi \\
  & \qquad + \| (k_n+1)\ \mathbf{X}^{k_n+1} - \mathscr{L}^{k_n}_{p,X}(\beta_i) \|_\pi \big) \\
	 & \leqslant  \frac{2}{\| \mathbf{X}^{k_n}\|_\pi} \|  \mathscr{L}^{k_n}_{p,X}(\beta_i) - (k_n+1)\ \mathbf{X}^{k_n+1}\|_\pi \\
	 & \leqslant 4 \ell e^{(M-1)K(\omega)}\Big(\frac{1}{\sqrt{k_n+1}} \sqrt{\frac{1 - (t_i - t_{i-1})}{t_i - t_{i-1}}} \\
  & \qquad \qquad \qquad \quad +  4\exp\Big(- \frac{k_n}{16}(t_i-t_{i-1})^2\Big) \Big).
\end{align*}
{\qed}
\end{proof}

\begin{remark}
    Note that Theorem \ref{th:insertion_cvg} holds for the projective norm. An extension of this theorem to other admissible norms is possible, provided a lower bound similar to Proposition \ref{th:sig_lower_bound} is available. However, the existence of such a bound is still a conjecture \citep[see][]{chang2018super}. On the other hand, computing the projective norm is computationally expensive, so in practice it is much better to solve \eqref{eq:min_problem_insertion_algo} with the Euclidean tensor norm, as we will do in the next subsection.
\end{remark}

\subsection{Algorithm}

The last step necessary to have a complete algorithm is to solve \eqref{eq:min_problem_insertion_algo} with the choice of the Euclidean tensor norm. Let $A_p \in \R^{d^{n+1} \times d}$ be the matrix representing the linear map $\mathscr{L}^n_{p,X}(\cdot)$ in the canonical basis of $\R^{d}$, where $(\R^d)^{\otimes n+1}$ is identified with $\R^{d^{n+1}}$. Thus, for any $y\in \R^d$, $\mathscr{L}^n_{p,X}(y)=A_p y$. %
The following lemma, which is essential for finding an explicit solution to \eqref{eq:min_problem_insertion_algo}, gives information about the form of the matrix $A_p^\top A_p$.
\begin{lemma}\label{lemma:singular_values_insertion_operator}
For any $p \in \{1, \dots, n+1 \}$, one has $A_p^\top A_p=\|\mathbf{X}^n \|^2 I_{d}$, where $I_{d}$ is the identity matrix.
\end{lemma}

Lemma \ref{lemma:singular_values_insertion_operator}, proved in the Supplementary Material, allows us to obtain an explicit solution to Problem \eqref{eq:min_problem_insertion_algo}.
\begin{proposition}
For any $p \in \{1, \dots, n+1 \}$, Problem \eqref{eq:min_problem_insertion_algo} with the Euclidean tensor norm has a unique solution, equal to 
\[ y^\ast_{p,n} = (n+1) \frac{A_p^\top  \mathbf{X}^{n+1}}{\| \mathbf{X}^n \|^2}. \]
\end{proposition}

\begin{proof}
Problem \eqref{eq:min_problem_insertion_algo} takes the form
\[y^\ast_{p,n} \in \argmin_{y \in \R^d} \|A_p y- (n+1)\mathbf{X}^{n+1}\|^2.\]
This is exactly the same minimization problem than in least-squares regression. Therefore,
\begin{align*}
y^\ast_{p,n} &= (n+1)(A_p^\top A_p)^{-1}A_p^\top \mathbf{X}^{n+1}.
\end{align*}
By Lemma \ref{lemma:singular_values_insertion_operator}, $A_p^\top A_p= \|\mathbf{X}^n\|^2 I_d$, and thus
\[y^\ast_{p,n} = (n+1)\frac{A_p^\top  \mathbf{X}^{n+1}}{\| \mathbf{X}^n \|^2}.\]
{\qed}
\end{proof}

\begin{algorithm}[ht]
\caption{Insertion algorithm}\label{alg:inversion}
\KwData{$d$: dimension of the underlying path; $n$: truncation order of the signature; $\mathbf{X}^n$ and $\mathbf{X}^{n+1}$: signature of order $n$; $X_0$: starting point of the path.}
\KwResult{$\widetilde{X} = (X_0, \widetilde{X}_{\nicefrac{1}{n}}, \dots,  \widetilde{X}_{1})\in \R^{(n+1) \times d }$: array of the $(n+1)$ positions in $\R^d$ of the reconstructed path.}

Extract $\mathbf{X}^n$ and $\mathbf{X}^{n-1}$, the terms of order $n$ and $n-1$ of $S^n(X)$.

$\widetilde{X}_0=X_0$

\For{$p \in \left\{1,\dots,n \right\}$}{
Compute $A_p$, the matrix representing the linear map $\mathscr{L}^{n-1}_{p,X}(\cdot)$, function of $\mathbf{X}^{n-1}$.

$y_{p,n}^*=n\frac{A_p^\top \mathbf{X}^n}{\|\mathbf{X}^{n-1}\|^2}$

$\widetilde{X}_{\nicefrac{p}{n}}=\widetilde{X}_{\nicefrac{p-1}{n}}+ \frac{y_{p,n}^\ast}{n}$

}
\end{algorithm}

The insertion algorithm is presented in Algorithm \ref{alg:inversion}. Given a signature of order $n \geqslant 2$, the finest piecewise approximation we can get is to let $p$ vary in $\{1, \dots, n\}$. For each such $p$, we solve \eqref{eq:min_problem_insertion_algo}, which gives an approximation of the slope $y^\ast_{p,n}$ on $[\nicefrac{(p-1)}{n},\nicefrac{p}{n}]$. We then integrate to obtain a piecewise-linear path $\tilde{X}$. Note that since we cannot recover the initial position of the path, it is provided by the user.

It is important to note that Theorem \ref{th:insertion_cvg} applies only to piecewise linear paths, but this is not a restriction per se, and Algorithm \ref{alg:inversion} is in fact applicable to any continuous path of bounded variation. To understand this remark, one only has to remember that in practice the path $X$ is never observed continuously, but only at discrete-time points $t_1, \hdots, t_M$, $0\leqslant t_1 < \cdots <t_M\leqslant 1$, with values $X_{t_1}, \ldots, X_{t_M}$. The signature computed by the softwares is not the signature of $X$, but the signature of the piecewise linear path obtained by linear interpolation of $X_{t_1}, \ldots, X_{t_M}$. This operation is performed sequentially using Chen's identity, as explained in the remark after Proposition \ref{prop:basic_sig_prop}, with a complexity $\mathcal{O}(M d^n)$. Thus, Algorithm \ref{alg:inversion} should rather be understood as a method to approximate any continuous path by a linear path composed of $n$ pieces, with two signatures close to each other. Note that it can be shown that the signatures of a given path and of its piecewise linear approximation converge as the discretization gets thinner \citep[see, e.g.,][Lemma 3.3]{bleistein2023learning}, which theoretically justifies working with piecewise linear approximations.

A natural question is at what order $n$ the path should be reconstructed. If $n$ is chosen too small, the quality of the approximation is of course likely to be poor, but if $n$ is too large, we run into memory problems, since the size of the signature is of the order of $d^n$. In practice, $n$ is chosen to be as large as possible, given the memory constraints. Finally, Algorithm \ref{alg:inversion} is computationally cheap, since the most expensive operation inside the loop is the matrix multiplication $A_p^\top \mathbf{X}^n$ (of the order $\mathcal{O}(d^{n+1})$). Moreover, the algorithm can be easily parallelized to allow it to take as input a batch of $N$ signatures, each of order $n$, and thus of length $\nicefrac{(d^{n+1}-1)}{(d-1)}$ for $d \geqslant 1$. In this parallel context, the data has the form of an array of size $\R^{N \times  \nicefrac{(d^{n+1}-1)}{(d-1)}}$ and the algorithm outputs a tensor of multiple paths of size $\R^{N\times  (n+1) \times d}$. We refer the reader to the code that we have included in the Signatory package \citep{signatory} for more details.

	\begin{figure}[ht]
		\centering	
		\begin{subfigure}[b]{0.33\textwidth}
			\includegraphics[width=\textwidth]{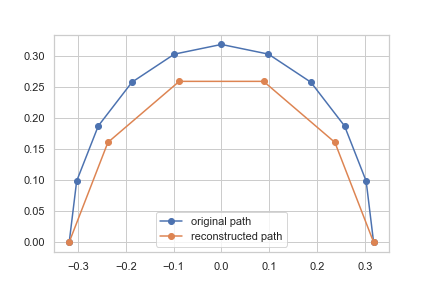}
			\caption{$n=5$}
			\label{fig:5_points_half_circle}
		\end{subfigure}%
		~
		\begin{subfigure}[b]{0.33\textwidth}
			\includegraphics[width=\textwidth]{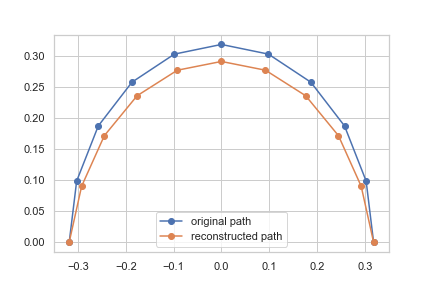}
			\caption{$n=10$}
			\label{fig:10_points_half_circle}
		\end{subfigure}%
		~
		\begin{subfigure}[b]{0.33\textwidth}
			\includegraphics[width=\textwidth]{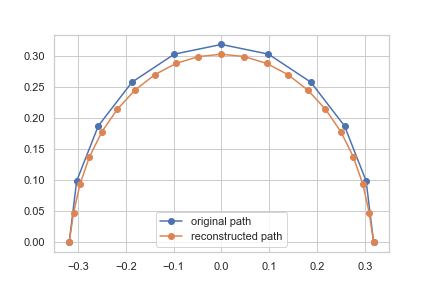}
			\caption{$n=20$}
			\label{fig:20_points_half_circle}
		\end{subfigure}%
		
		\caption{Signature inversion of a half circle.}
		\label{fig:half_circle_sig_inversion}
	\end{figure}

		\begin{figure}[ht]
		\centering	
		\begin{subfigure}[b]{0.33\textwidth}
			\includegraphics[width=\textwidth]{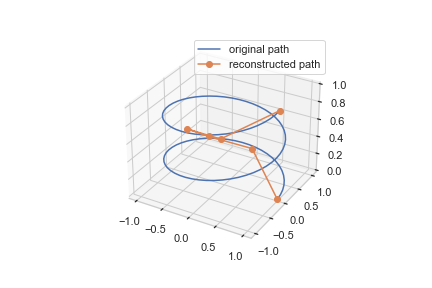}
			\caption{$n=5$}
			\label{fig:5_spiral}
		\end{subfigure}%
		~
		\begin{subfigure}[b]{0.33\textwidth}
			\includegraphics[width=\textwidth]{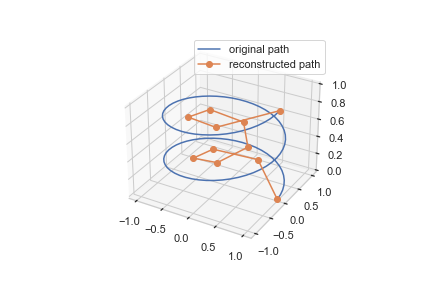}
			\caption{$n=10$}
			\label{fig:10_spiral}
		\end{subfigure}%
		~
		\begin{subfigure}[b]{0.33\textwidth}
			\includegraphics[width=\textwidth]{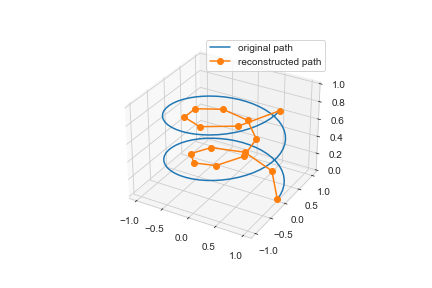}
			\caption{$n=15$}
			\label{fig:15_spiral}
		\end{subfigure}%
		
		\caption{Signature inversion of a spiral in 3d.}
		\label{fig:spiral_sig_inversion}
	\end{figure}

		\begin{figure}[ht]
		\centering
		\begin{subfigure}[b]{0.33\textwidth}
			\includegraphics[width=\textwidth]{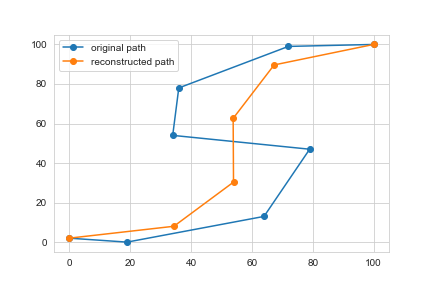}
			\caption{$n=5$}
			\label{fig:5_pendigit_5}
		\end{subfigure}%
		~
		\begin{subfigure}[b]{0.33\textwidth}
			\includegraphics[width=\textwidth]{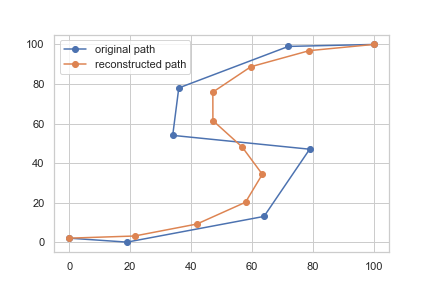}
			\caption{$n=10$}
			\label{fig:10_pendigit_5}
		\end{subfigure}%
		~
		\begin{subfigure}[b]{0.33\textwidth}
			\includegraphics[width=\textwidth]{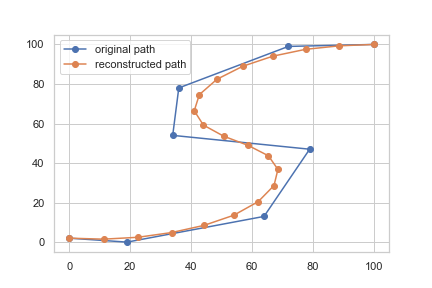}
			\caption{$n=20$}
			\label{fig:20_pendigit_5}
		\end{subfigure}%

		\caption{Signature inversion of an example of the class "5" from the Pendigits dataset.}
		\label{fig:pendigit_5_sig_inversion}
	\end{figure}

	\begin{figure}[ht]
		\centering	
		\begin{subfigure}[b]{0.33\textwidth}
			\includegraphics[width=\textwidth]{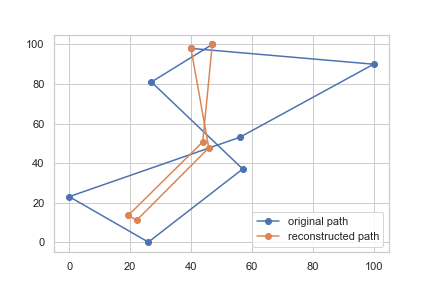}
			\caption{$n=5$}
			\label{fig:5_pendigit_8}
		\end{subfigure}%
		~
		\begin{subfigure}[b]{0.33\textwidth}
			\includegraphics[width=\textwidth]{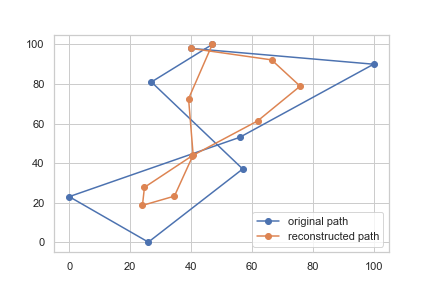}
			\caption{$n=10$}
			\label{fig:10_pendigit_8}
		\end{subfigure}%
		~
		\begin{subfigure}[b]{0.33\textwidth}
			\includegraphics[width=\textwidth]{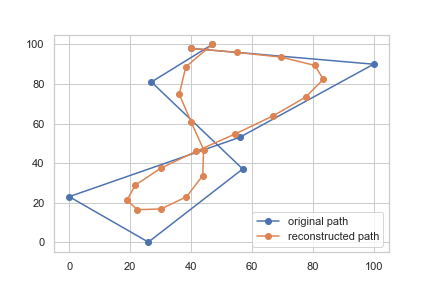}
			\caption{$n=20$}
			\label{fig:20_pendigit_8}
		\end{subfigure}%
		
		\caption{Signature inversion of an example of the class "8" from the Pendigits dataset.}
		\label{fig:pendigit_8_sig_inversion}
	\end{figure}
\section{Experimental results} \label{sec:experimental_results}

\subsubsection*{Examples of inversion}
We show in Figures \ref{fig:half_circle_sig_inversion}, \ref{fig:spiral_sig_inversion}, \ref{fig:pendigit_5_sig_inversion}, and \ref{fig:pendigit_8_sig_inversion} some examples of reconstruction for paths in dimension 2 and 3. Figures \ref{fig:half_circle_sig_inversion} and \ref{fig:spiral_sig_inversion} are simulations of a half circle and a spiral, while Figures \ref{fig:pendigit_5_sig_inversion} and \ref{fig:pendigit_8_sig_inversion} are two examples of the Pendigits dataset from the UCI Machine Learning Repository \citep{dua2019}. In these four figures, we have knowledge of the ``true'' path, in blue, from which we compute its truncated signature with three different values of $n$. We then invert this signature using Algorithm \ref{alg:inversion} and plot the reconstructed path in orange. This way we can compare the quality of the reconstructed path with the original one. We observe that in all these experiments, the reconstructed paths are smooth and close to the original ones for $n=20$.

	\begin{figure}[ht]
		\centering	
		\begin{subfigure}[b]{0.33\textwidth}
			\includegraphics[width=\textwidth]{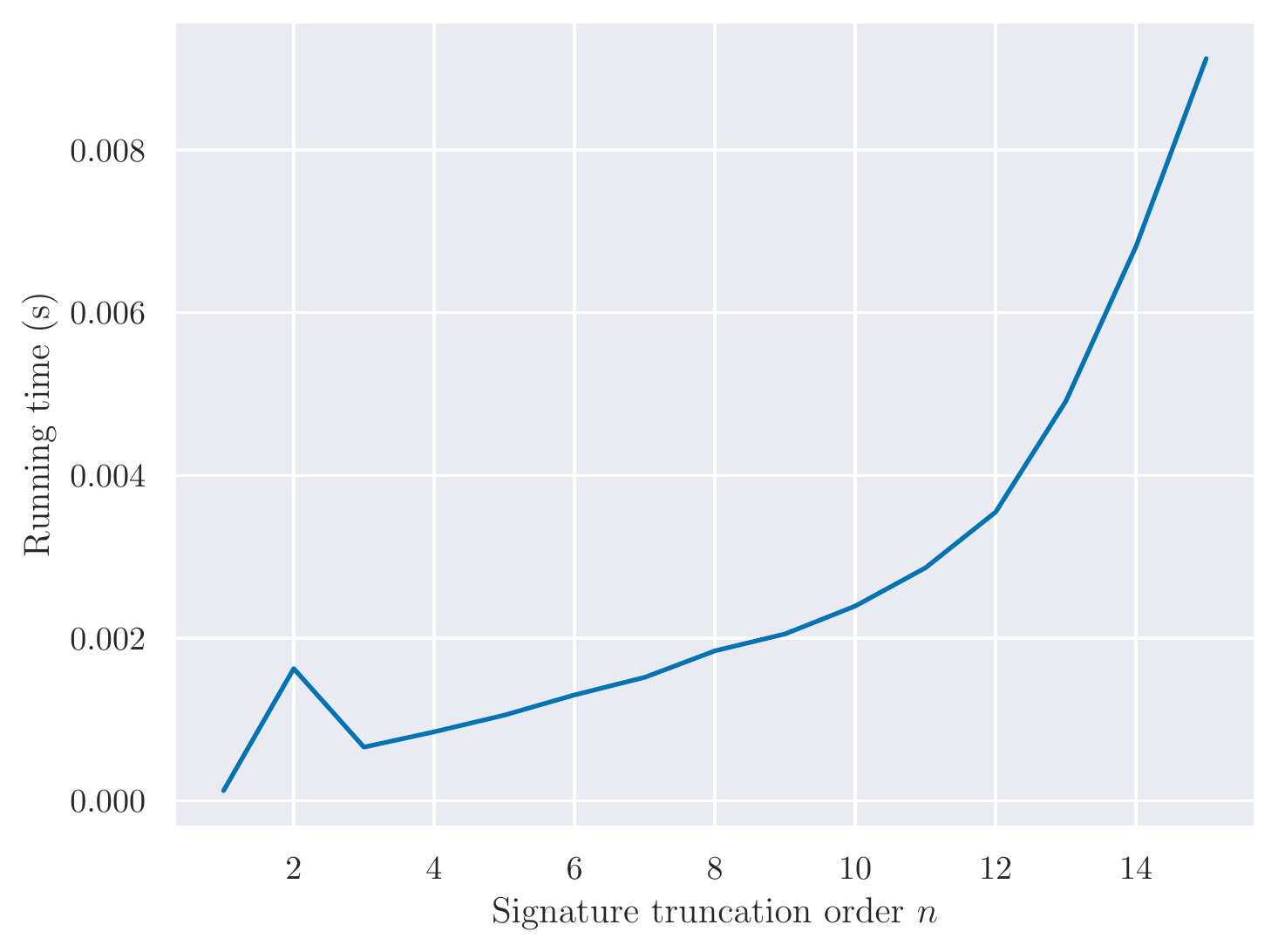}
			\caption{Running time when $n$ varies.}
			\label{fig:running_time_depth}
		\end{subfigure}%
		~
		\begin{subfigure}[b]{0.33\textwidth}
			\includegraphics[width=\textwidth]{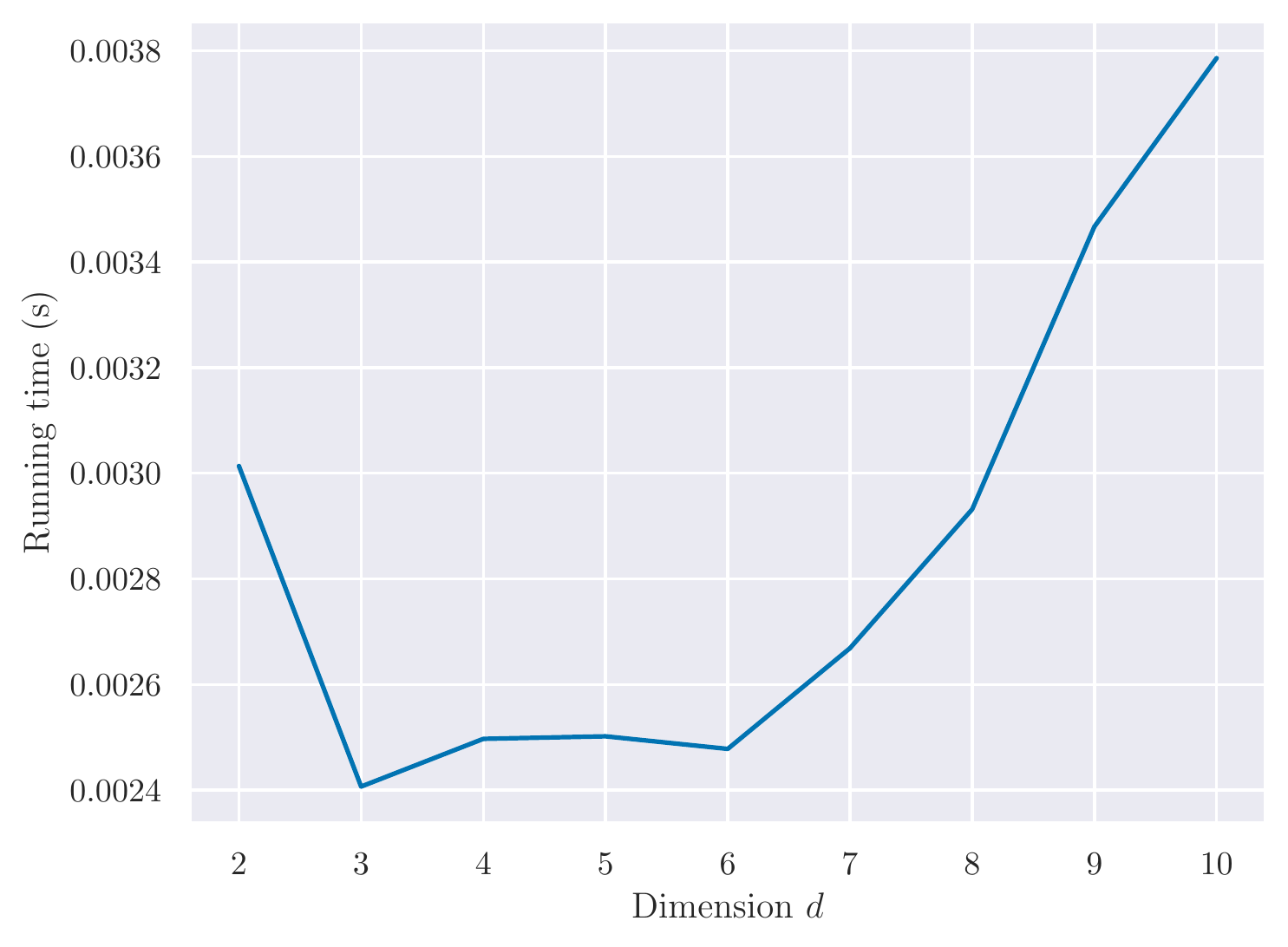}
			\caption{Running time when $d$ varies.}
			\label{fig:running_time_channels}
		\end{subfigure}%
		~
		\begin{subfigure}[b]{0.33\textwidth}
			\includegraphics[width=\textwidth]{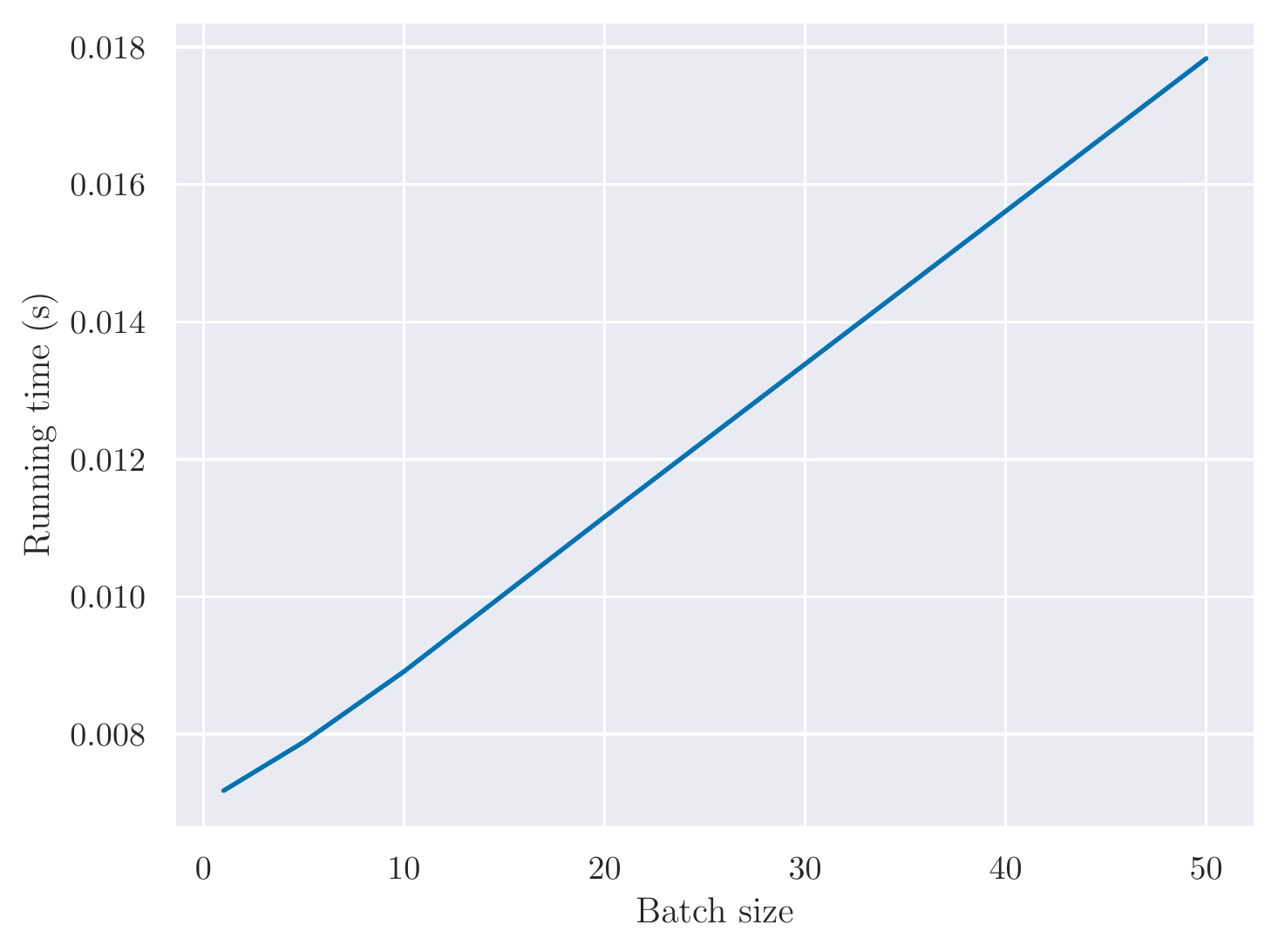}
			\caption{Running time when $N$ varies.}
			\label{fig:running_time_batch}
		\end{subfigure}%
		
		\caption{Running time in seconds to invert the signature of several paths, when we let some hyperparameters vary: the truncation order of the signature $n$, the dimension of the path $d$, and the number of signatures that are inverted.}
		\label{fig:running_times}
	\end{figure}

\subsubsection*{Running times} 
Figure \ref{fig:running_times} shows the running time of Algorithm \ref{alg:inversion} when several hyperparameters vary. The paths are randomly generated as piecewise linear paths with $M=10$ pieces, starting at 0. The endpoint of each linear piece is generated uniformly at random in $[0,1]^d$. The experiments were performed on a standard laptop computer.

In Figure \ref{fig:running_time_depth}, we invert one path in dimension $d=2$ for different values of $n$. In Figure \ref{fig:running_time_channels}, we set $n=4$ and let $d$ vary. In Figure \ref{fig:running_time_batch}, we set $n=10$ and $d=2$, and let the number of paths $N$ vary. We see that as $n$ increases, the running time increases exponentially, while the dependence on the dimension $d$ is polynomial, and the dependence on the number of paths $N$ is linear. In total, we can invert 50 signatures truncated at order 10 in about 0.02s.

	\begin{figure}[ht]
		\centering	
  		\begin{subfigure}[b]{0.45\textwidth}
			\includegraphics[width=\textwidth]{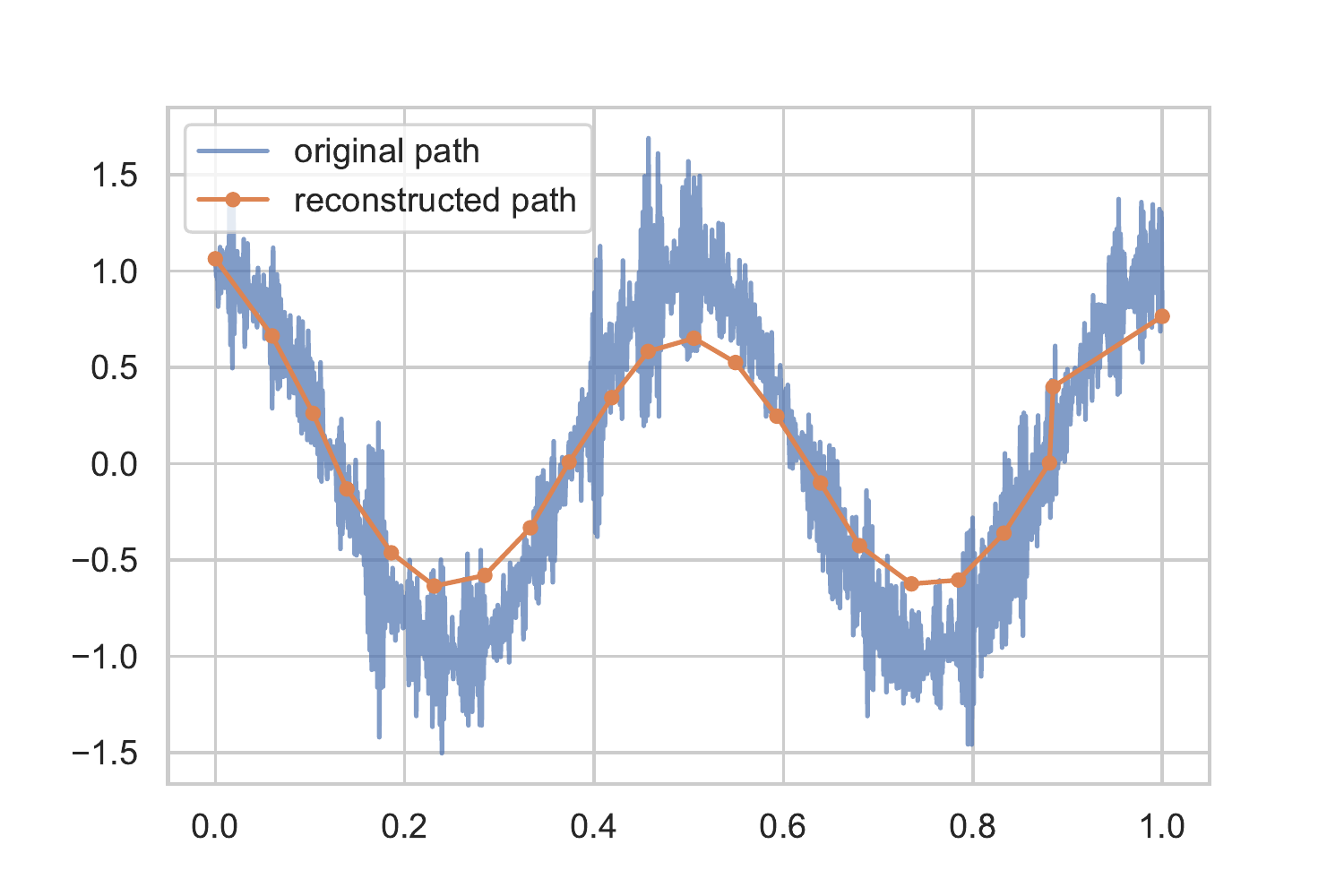}
			\caption{Simulation of a path as a cosine function plus an AR(1) noise.}
		\end{subfigure}%
  		~\hfill
		\begin{subfigure}[b]{0.45\textwidth}
			\includegraphics[width=\textwidth]{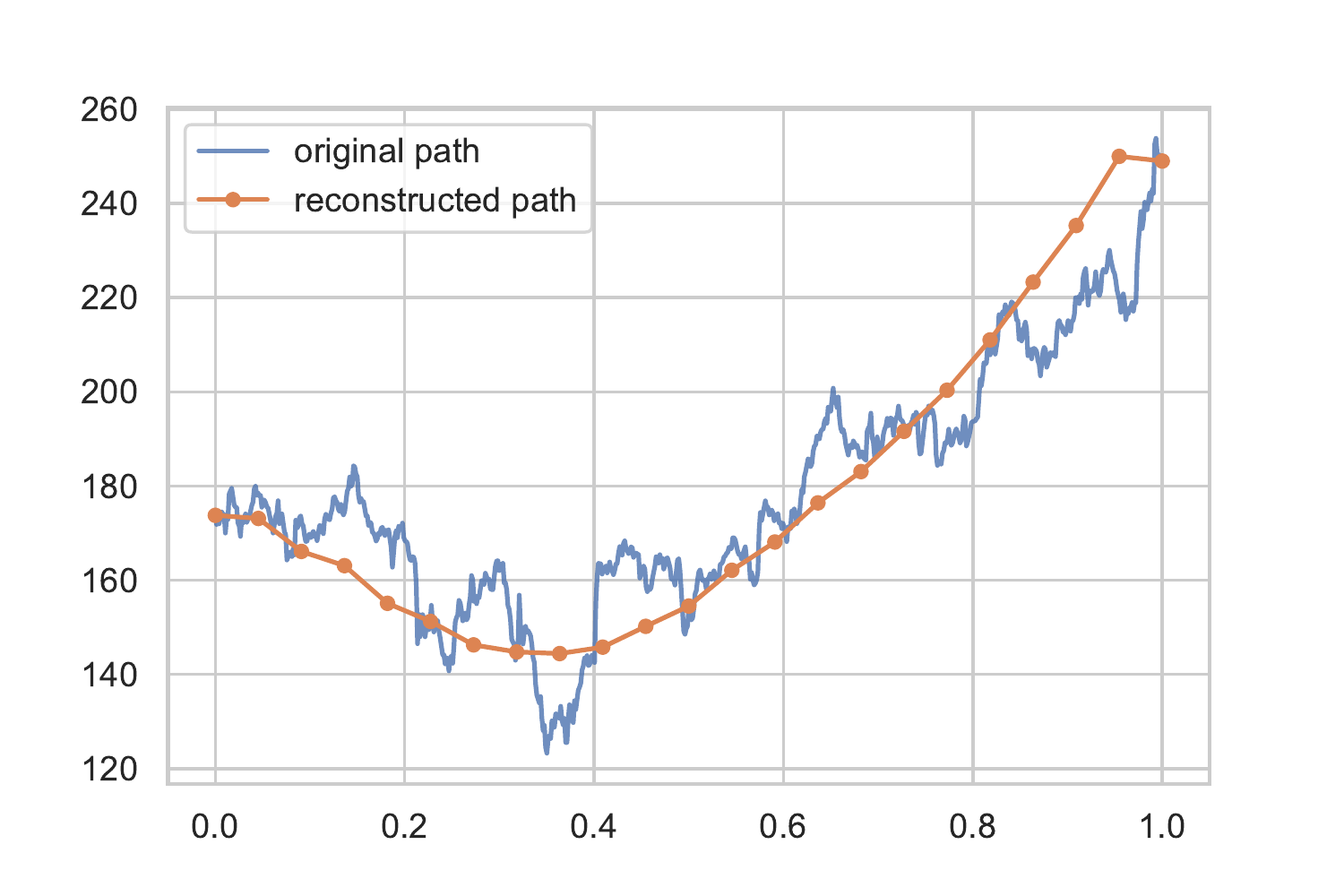}
			\caption{Daily open value of the FDX between 2015 and 2017.}
		\end{subfigure}%
		\caption{Examples of trend estimation: the signature of the original path in blue is computed with a truncation order of $n=21$, then it is inverted, which gives the orange reconstructed path.}
		\label{fig:trend_estimation}
	\end{figure}

\subsubsection*{Trend estimation} 

One interesting application of Algorithm \ref{alg:inversion} is to perform nonparametric trend estimation. In fact, if we compute the signature of a highly-sampled time series and invert it, we get an approximation of the original path sampled at a much lower frequency. In Figure \ref{fig:trend_estimation}, we show some results on both simulated and real-world financial data of this approach. We can see that the signatures are able to capture the overall trend of the curve. In this case, the choice of truncation order $n$ corresponds to the degree of smoothness desired for the trend, which can then be adapted to the application at hand. An interesting extension of our work would be to study both theoretically and empirically the resulting estimator of the trend, and compare it with traditional approaches such as running average. We can expect the signature-based estimator to be less sensitive to outliers and noise.

\begin{acknowledgement}
Terry Lyons was funded in part by the EPSRC [grant number EP/S026347/1], in part by The Alan Turing Institute under the EPSRC grant EP/N510129/1, the Data Centric Engineering Programme (under the Lloyd’s Register Foundation grant G0095), the Defence and Security Programme (funded by the UK Government) and the Office for National Statistics \& The Alan Turing Institute (strategic partnership), and in part by the Hong Kong Innovation and Technology Commission (InnoHK Project CIMDA). Adeline Fermanian thanks Linus Bleistein and Claire Boyer for stimulating discussions. The authors thank Davy Paindaveine and two anonymous referees for their valuable comments on a first version of the paper.
\end{acknowledgement}
\bibliographystyle{plainnat}
\bibliography{references}

\newpage
\begin{center}
  {\Large \textbf{Supplementary Material}}  
\end{center}
\section{Proof of Proposition \ref{thm:insertion_X'_theta_cvg}}
We start by proving a relation similar to Chen's identity (Proposition \ref{prop:basic_sig_prop}, $(i)$) for the insertion operator. For any $X \in BV(\R^d)$, $n \geqslant 1$, and $p \in \{1, \dots, n+1\}$, the insertion operator restricted to an interval $[s,t] \subset [0,1]$ is defined by
\begin{equation*}
	\mathscr{L}^n_{p,X; [s,t]}(y) = \int_{(u_1, \dots, u_n) \in \Delta_{n;[s,t]}} dX_{u_1} \otimes \dots \otimes dX_{u_{p-1}} \otimes y \otimes dX_{u_{p}} \otimes \dots \otimes dX_{u_n}.
\end{equation*}

\begin{lemma}
\label{lemma:chen_insertion_operator}
Let $X \in BV(\R^d)$, $n \geqslant 1$, $p \in \{1, \dots, n+1\}$, and $u, v, w \in  [0,1]$ such that $u<v<w$. Then
\begin{equation*}
\mathscr{L}^n_{p,X; [u,w]}(y) = \sum_{k=0}^{p-1} \mathbf{X}^k_{[u,v]} \otimes \mathscr{L}^{n-k}_{p-k, X; [v,w]}(y) + \sum_{k=p}^n  \mathscr{L}^k_{p, X; [u,v]}(y) \otimes \mathbf{X}^{n-k}_{[v,w]}.
\end{equation*}
\end{lemma}
\begin{proof}
	This formula is obtained by splitting the integration domain in two parts, with a running index $k$ corresponding to the number of integration variables in each part of the interval. We have
	\begin{align*}
	&\mathscr{L}^n_{p,X; [u,w]}(y) \\
    & = \idotsint\limits_{u \leqslant u_1 \leqslant \dots \leqslant u_n \leqslant w}dX_{u_1} \otimes \dots \otimes dX_{u_{p-1}} \otimes y \otimes \dots \otimes dX_{u_n} \\
	& = \sum_{k=0}^n \quad \idotsint\limits_{u \leqslant u_1 \leqslant \dots \leqslant u_k \leqslant v \leqslant u_{k+1} \leqslant \dots \leqslant u_n \leqslant w}dX_{u_1} \otimes \dots \otimes dX_{u_{p-1}} \otimes y \otimes \dots \otimes dX_{u_n}\\
	& = \! \sum_{k=0}^n \int_{(u_1, \dots, u_k) \in \Delta_{k;[u,v]}} \! \int_{(u_{k+1}, \dots, u_n) \in \Delta_{n-k; [v,w]}} \! dX_{u_1} \otimes \dots \otimes dX_{u_{p-1}} \otimes y \otimes \dots \otimes dX_{u_n},
	\end{align*}
	where we use the convention $u_0=u$ and $u_{n+1}=w$. To simplify this expression, we analyze separately the case $k \leqslant p-1$ and the case $k > p$. So,
	\begin{align*}
		&\mathscr{L}^n_{p,X; [u,w]}(y) \\
		& \quad = \sum_{k=0}^{p-1} \int_{\Delta_{k;[u,v]}} \int_{\Delta_{n-k; [v,w]}} dX_{u_1} \otimes \dots \otimes dX_{u_{p-1}} \otimes y \otimes \dots \otimes dX_{u_n}  \\
		& \qquad + \sum_{k=p}^n\int_{\Delta_{k;[u,v]}} \int_{\Delta_{n-k; [v,w]}}  dX_{u_1} \otimes \dots \otimes dX_{u_{p-1}} \otimes y \otimes \dots \otimes dX_{u_n} \\
		& \quad = \sum_{k=0}^{p-1} \Big( \int_{(u_1, \dots, u_k) \in \Delta_{k;[u,v]}} dX_{u_1} \otimes \dots \otimes dX_{u_{k}} \Big) \\
		& \qquad \qquad  \otimes \Big(  \int_{(u_{k+1}, \dots, u_n) \in \Delta_{n-k; [v,w]}} dX_{u_{k+1}} \otimes \dots \otimes dX_{u_{p-1}} \otimes y \otimes \dots \otimes dX_{u_n} \Big) \\
		& \qquad  + \sum_{k=p}^n \Big( \int_{(u_1, \dots, u_k) \in \Delta_{k;[u,v]}}dX_{u_1} \otimes \dots \otimes dX_{u_{p-1}} \otimes y \otimes \dots \otimes dX_{u_{k}} \Big) \\
		&\quad \qquad \qquad \otimes \Big( \int_{(u_{k+1}, \dots, u_n) \in \Delta_{n-k; [v,w]}} dX_{u_{k+1}} \otimes \dots \otimes dX_{u_n} \Big) \\
		&\quad  = \sum_{k=0}^{p-1} \mathbf{X}^k_{[u,v]} \otimes \mathscr{L}^{n-k}_{p-k, X; [v,w]}(y) + \sum_{k=p}^n  \mathscr{L}^k_{p, X; [u,v]}(y) \otimes \mathbf{X}^{n-k}_{[v,w]}.
	\end{align*}
{\qed}
\end{proof}

We are now ready to prove Proposition \ref{lemma:chen_insertion_operator}. The first step uses Proposition \ref{prop:basic_sig_prop}, $(i)$ and Lemma \ref{lemma:chen_insertion_operator} to split both the signature $\mathbf{X}^{n+1}$ and the insertion operator $\mathscr{L}^n_{p,X}(y)$ on the intervals $[0, t_{i-1}]$, $[t_{i-1}, t_i]$, and $[t_i, 1]$. With respect to the insertion operator, we split the sum depending on the location of $p$, similar to the proof of Lemma \ref{lemma:chen_insertion_operator}. Let $\mathscr{S}$ denote the set
\begin{equation*}
\mathscr{S} = \big\{(n_1, n_2, n_3) \in \{0, \dots, n+1\} \mid n_1 + n_2 + n_3 = n+1 \big\}.
\end{equation*}
For a fixed $p \in \{1, \dots, n+1 \}$, we use the following partition, illustrated in Figure \ref{fig:partition_chen_3_intervals}:
\begin{align*}
\mathscr{S} &= \big\{(n_1, n_2, n_3) \in \mathscr{S} \mid n_1 \geqslant p \big\} \cup \big\{(n_1, n_2, n_3) \in \mathscr{S} \mid n_1 < p, \, n_1 + n_2 \geqslant p \big\} \\
& \quad \cup \big\{(n_1, n_2, n_3) \in \mathscr{S} \mid n_1 + n_2 < p \big\}.
\end{align*}
\colorlet{lightgray}{gray!10}
\definecolor{colorp1}{HTML}{0173b2}
\definecolor{colorp2}{HTML}{de8f05}
\definecolor{colorp3}{HTML}{029e73}

\global\def\plotrange{0, 1, 4, 8, 11, 15, 18, 19}
\def\myarr {
    (0.131579, 0.5)
    (0.2631579, 1.228486295)
    (1.05263158, 2.37243038 - 0.3)
    (2.10526316, 2.232905675 - 0.3)
    (2.89473684 - 0.1, 1.597253615)
    (3.94736842 - 0.3, 1.09415549 + 0.2)
    (4.7368421 - 0.5, 1.6927394599999999)
    (5.0 -0.3, 2.5)
}

\global\def\xs{{0.131579, 0.2631579, 0.52631578, 0.78947368, 1.05263158, 1.31578948, 1.57894736, 1.84210526, 2.10526316, 2.36842106, 2.63157894, 2.89473684 - 0.1, 3.15789474, 3.42105264, 3.68421052, 3.94736842 - 0.3, 4.21052632, 4.47368422, 4.7368421 - 0.5, 5.0 - 0.3}}

\global\def\ys{{0.5, 1.228486295, 1.7697186200000001, 2.14419923, 2.37243038 - 0.3, 2.47491434 - 0.3, 2.47215338, 2.384649725, 2.232905675 - 0.3, 2.037423455, 1.81870535, 1.597253615, 1.39357049, 1.2281582599999998, 1.121519165, 1.09415549 + 0.2, 1.166569475, 1.359263375, 1.6927394599999999, 2.5}}

\global\def\xlen{19}

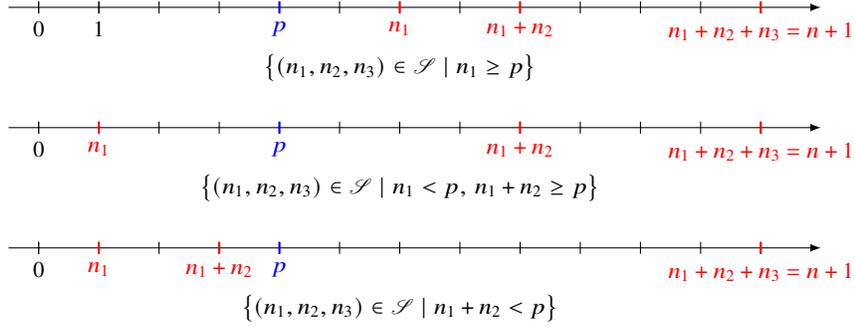
\begin{figure}[ht]
\centering
    \begin{tikzpicture}[scale=0.8]
      \draw[black, ->] (-0.5, 0) -- (13, 0);

      \foreach \x in {0,1,...,12} {
      \draw (\x,0.1) -- (\x,-0.1);
      }
      
      \draw (0,0.1) -- (0,-0.1) node[below] {$0$};
      \draw[thick, red] (12,0.1) -- (12,-0.1) node[below] {$n_1 + n_2 +n_3 = n+1$};
      \draw[thick, blue] (4,0.1) -- (4,-0.1) node[below] {$p$};
      \draw[thick, red] (1,0.1) -- (1,-0.1) node[below] {$n_1$};
      \draw[thick, red] (3,0.1) -- (3,-0.1) node[below] {$n_1 + n_2$};

      \draw (6, -1) node{$\big\{(n_1, n_2, n_3) \in \mathscr{S} \mid n_1 + n_2 < p \big\}$};

      \draw[black, ->] (-0.5, 2) -- (13, 2);

      \foreach \x in {0,1,...,12} {
      \draw (\x,2.1) -- (\x,1.9);
      }
      \draw (0,2.1) -- (0,1.9) node[below] {$0$};
      \draw[thick, red] (12,2.1) -- (12,1.9) node[below] {$n_1 + n_2 +n_3 = n+1$};
      \draw[thick, blue] (4,2.1) -- (4,1.9) node[below] {$p$};
      \draw[thick, red] (1,2.1) -- (1,1.9) node[below] {$n_1$};
      \draw[thick, red] (8,2.1) -- (8,1.9) node[below] {$n_1 + n_2$};

      \draw (6, 1) node{$\big\{(n_1, n_2, n_3) \in \mathscr{S} \mid n_1 < p, \, n_1 + n_2 \geq p \big\}$};

      \draw[black, ->] (-0.5, 4) -- (13, 4);

      \foreach \x in {0,1,...,12} {
      \draw (\x,4.1) -- (\x,3.9);
      }
      
      \draw (0,4.1) -- (0,3.9) node[below] {$0$};
       \draw (1,4.1) -- (1,3.9) node[below] {$1$};
      \draw[thick, red] (12,4.1) -- (12,3.9) node[below] {$n_1 + n_2 +n_3 = n+1$};
      \draw[thick, blue] (4,4.1) -- (4,3.9) node[below] {$p$};
      \draw[thick, red] (6,4.1) -- (6,3.9) node[below] {$n_1$};
      \draw[thick, red] (8,4.1) -- (8,3.9) node[below] {$n_1 + n_2$};

      \draw (6, 3) node{$\big\{(n_1, n_2, n_3) \in \mathscr{S} \mid n_1 \geq p \big\}$};

    \end{tikzpicture}

\caption{Partition of $\mathscr{S}$.}
\label{fig:partition_chen_3_intervals}
\end{figure}
With this notation, we have
\begin{align*} 
 \mathscr{L}^n_{p,X}(y) & = \sum_{\substack{(n_1, n_2, n_3 ) \in \mathscr{S}\\ p \leqslant n_1 }} \mathscr{L}^{n_1 -1}_{p, X; [0,t_{i-1}]}(y) \otimes \mathbf{X}^{n_2}_{[t_{i-1}, t_i]} \otimes \mathbf{X}^{n_3}_{[t_i, 1]}  \\
 & \quad + \sum_{\substack{(n_1, n_2, n_3 ) \in \mathscr{S} \\ n_1 < p \leqslant n_1 + n_2 }} \mathbf{X}^{n_1}_{[0, t_{i-1}]} \otimes \mathscr{L}^{n_2 - 1}_{p - n_1, X; [t_{i-1},t_{i}]}(y) \otimes \mathbf{X}^{n_3}_{[t_i, 1]} \\
& \quad +  \sum_{\substack{(n_1, n_2, n_3 ) \in \mathscr{S} \\ n_1 + n_2 < p }} \mathbf{X}^{n_1}_{[0, t_{i-1}]} \otimes \mathbf{X}^{n_2}_{[t_{i-1}, t_i]}\otimes \mathscr{L}^{n_3 -1}_{p - n_1 - n_2 , X; [t_{i},1]}(y)  \\
& = \sum_{k=p-1}^{n } \mathscr{L}^{k}_{p, X; [0,t_{i-1}]}(y) \otimes \mathbf{X}^{n-k}_{[t_{i-1},1]}   \\
& \quad + \sum_{\substack{(n_1, n_2, n_3 ) \in \mathscr{S} \\ n_1 < p \leqslant n_1 + n_2 }} \mathbf{X}^{n_1}_{[0, t_{i-1}]} \otimes \mathscr{L}^{n_2 - 1}_{p - n_1, X; [t_{i-1},t_{i}]}(y) \otimes \mathbf{X}^{n_3}_{[t_i, 1]}  \\
& \quad + \sum_{k=0}^{p-1} \mathbf{X}^k_{[0, t_{i}]} \otimes \mathscr{L}^{n-k}_{p - k , X; [t_{i},1]}(y). 
\end{align*}
With respect to $\mathbf{X}^{n+1}$, a straightforward extension of Proposition \ref{prop:basic_sig_prop} yields
\begin{align*} 
\mathbf{X}^{n+1} &= \sum_{(n_1, n_2, n_3 ) \in \mathscr{S}} \mathbf{X}^{n_1}_{[0,t_{i-1}]} \otimes \mathbf{X}^{n_2}_{[t_{i-1},t_i]} \otimes \mathbf{X}^{n_3}_{[t_i,1]},
\end{align*}
and thus
\begin{align*} 
 \mathbf{X}^{n+1} & = \sum_{k=p-1}^{n} \mathbf{X}^{k+1}_{[0, t_{i-1}]} \otimes \mathbf{X}^{n+1-k}_{[t_{i-1}, 1]} + \sum_{\substack{(n_1, n_2, n_3 ) \in \mathscr{S} \\ n_1 < p \leqslant n_1 + n_2 }} \mathbf{X}^{n_1}_{[0, t_{i-1}]} \otimes \mathbf{X}^{n_2}_{[t_{i-1},t_{i}]} \otimes \mathbf{X}^{n_3}_{[t_i, 1]} \nonumber \\
 & \quad + \sum_{k=0}^{p-1} \mathbf{X}^k_{[0,t_i]} \otimes \mathbf{X}^{n+1-k}_{[t_i,1]}.
\end{align*}
Putting everything together, we are led to
\begin{align}\label{eq:chen_decomposition_diff_insertion_sig}
&n!(\mathscr{L}^n_{p,X} (\beta_i) - (n+1)\mathbf{X}^{n+1}) \nonumber \\
 & \quad = n! \sum_{k=p - 1 }^{n} (\mathscr{L}^{k}_{p,X; [0,t_{i-1}]}(\beta_i) - (n+1) \mathbf{X}^{k+1}_{[0, t_{i-1}]}) \otimes \mathbf{X}^{n-k}_{[t_{i-1},1]} \nonumber \\
& \qquad+ n!  \sum_{\substack{(n_1, n_2, n_3 ) \in \mathscr{S} \\ n_1 < p \leqslant n_1 + n_2 }} \mathbf{X}^{n_1}_{[0, t_{i-1}]} \otimes (\mathscr{L}^{n_2-1}_{p - n_1,X; [t_{i-1},t_{i}]}(\beta_i) - (n+1) \mathbf{X}^{n_2}_{[t_{i-1}, t_i]})\otimes \mathbf{X}^{n_3}_{[t_i, 1]} \nonumber \\
& \qquad + n! \sum_{k=0}^{p-1} \mathbf{X}^k_{[0, t_{i}]} \otimes (\mathscr{L}^{n-k}_{p - k,X; [t_{i},1]}(\beta_i) - (n+1) \mathbf{X}^{n+1 -k}_{[t_i, 1]}) \nonumber \\
& \quad := A_1 + A_2 + A_3.
\end{align}
We now bound $A_1$, $A_2$, and $A_3$ separately. Note that the same arguments as in the proof of Proposition \ref{insertiontest} give for any admissible norm $\|\mathscr{L}^n_{p,X;[s,t]}(y) \|= \| \mathbf{X}^n\| \cdot \|y\|$ and thus, by  Proposition \ref{prop:up_bound_norm_sig},
\[\| \mathscr{L}^n_{p,X;[s,t]}(y) \| \leqslant \frac{\|X\|_{TV;[s,t]}^n}{n!} \|y\|.\]
This yields, together with the triangle inequality, 
\begin{align*}
\|A_1\| & = \Big\| n! \sum_{k=p - 1 }^{n} (\mathscr{L}^{k}_{p,X; [0,t_{i-1}]}(\beta_i) - (n+1) \mathbf{X}^{k+1}_{[0, t_{i-1}]}) \otimes \mathbf{X}^{n-k}_{[t_{i-1},1]}\Big\| \\
& \leqslant n! \sum_{k=p - 1 }^{n} \big\| \mathscr{L}^{k}_{p,X; [0,t_{i-1}]}(\beta_i) \big\| \cdot \big\| \mathbf{X}^{n-k}_{[t_{i-1},1]}  \big\| \\
& \quad + (n+1)! \sum_{k=p - 1 }^{n} \big\| \mathbf{X}^{k+1}_{[0, t_{i-1}]} \big\| \cdot \big\| \mathbf{X}^{n-k}_{[t_{i-1},1]}  \big\| \\
& \leqslant n! \sum_{k=p - 1 }^{n} \|\beta_i\| \frac{\|X\|_{TV; [0,t_{i-1}]}^{k}}{k!} \frac{\|X\|_{TV; [t_{i-1}, 1]}^{n-k}}{(n-k)!}\\
& \quad + (n+1)! \sum_{k=p - 1 }^{n} \frac{\|X\|_{TV; [0,t_{i-1}]}^{k+1}}{(k+1)!} \frac{\|X\|_{TV; [t_{i-1}, 1]}^{n-k}}{(n-k)!}.
\end{align*}
Recall that our choice of parameterization of $X$ ensures that for any $j \in \{1, \dots, M\}$, $\| \beta_j\| = \ell$. Therefore,
\[ \|X\|_{TV; [0,t_{i-1}]} = \sum_{j=1}^{i-1} \| \beta_j\| (t_j - t_{j-1}) = \ell \sum_{j=1}^{i-1}(t_j - t_{j-1}) = \ell t_{i-1}.\]
Similarly, $\|X\|_{TV; [t_{i-1}, 1]} = \ell (1-t_{i-1})$. It follows that
\begin{align}\label{eq:result_bound_A_1}
\|A_1\| & \leqslant \ell^{n+1}\bigg( \sum_{k=p - 1 }^{n} \binom{n}{k} t_{i-1}^k (1-t_{i-1})^{n-k} + \sum_{k=p}^{n+1} \binom{n+1}{k} t_{i-1}^k (1-t_{i-1})^{n+1-k} \bigg).
\end{align}
Bounding $A_3$ in a similar way, we obtain
\begin{align}\label{eq:result_bound_A_3}
\| A_3\| & \leqslant \ell^{n+1} \bigg(\sum_{k=0}^{p-1} \binom{n}{k} t_{i}^k (1-t_{i})^{n-k} + \sum_{k=0}^{p-1} \binom{n+1}{k} t_{i}^k (1-t_{i})^{n+1-k} \bigg).
\end{align}
We now turn to the term $A_2$. We have
\begin{align*}
\|A_2\| 
& \leqslant n! \! \sum_{\substack{(n_1, n_2, n_3 ) \in \mathscr{S} \\ n_1 < p \leqslant n_1 + n_2 }} \big\|\mathbf{X}^{n_1}_{[0, t_{i-1}]} \big\| \, \big\| \mathscr{L}^{n_2-1}_{p - n_1,X; [t_{i-1},t_{i}]}(\beta_i) - (n+1) \mathbf{X}^{n_2}_{[t_{i-1}, t_i]} \big\| \, \big\| \mathbf{X}^{n_3}_{[t_i, 1]} \big\| \\
& \leqslant \!  n! \sum_{(n_1, n_2, n_3 ) \in \mathscr{S}} \frac{t_{i-1}^{n_1} (1 - t_{i})^{n_3} \ell^{n_1} \ell^{n_3} }{n_1!n_3!} \big\| \mathscr{L}^{n_2-1}_{p - n_1,X; [t_{i-1},t_{i}]}(\beta_i) - (n+1) \mathbf{X}^{n_2}_{[t_{i-1}, t_i]}\big \|.
\end{align*}
Since $X$ is linear on $[t_{i-1}, t_i]$,
\begin{equation*}
\mathbf{X}^{n_2}_{[t_{i-1}, t_i]} = \frac{(t_i -t_{i-1})^{n_2}}{n_2!} \beta_i^{\otimes n_2}, \quad \mathscr{L}^{n_2-1}_{p - n_1,X; [t_{i-1},t_{i}]}(\beta_i) = \frac{(t_i -t_{i-1})^{n_2-1}}{(n_2-1)!} \beta_i^{\otimes n_2}.
\end{equation*}
So,
\begin{align*}
\big \| \mathscr{L}^{n_2-1}_{p - n_1,X; [t_{i-1},t_{i}]}(\beta_i) - (n+1) \mathbf{X}^{n_2}_{[t_{i-1}, t_i]}\big \| 
&= \frac{(t_i -t_{i-1})^{n_2}}{n_2!} \Big|\frac{n_2}{t_i - t_{i-1}} -(n+1) \Big| \cdot \|\beta_i^{\otimes n_2}\| \nonumber \\
& = \frac{(t_i -t_{i-1})^{n_2} \ell^{n_2}}{n_2!} \Big|\frac{n_2}{t_i - t_{i-1}} - (n+1) \Big|,
\end{align*}
and
\begin{align}\label{eq:result_bound_A_2}
& \|A_2\| \nonumber \\
& \quad \leqslant  n! \! \sum_{(n_1, n_2, n_3 ) \in \mathscr{S}} \frac{t_{i-1}^{n_1} \ell^{n_1}}{n_1!} \frac{(t_i -t_{i-1})^{n_2} \ell^{n_2}}{n_2!} \Big|\frac{n_2}{t_i - t_{i-1}} - (n+1) \Big| \frac{(1 - t_{i})^{n_3}\ell^{n_3}}{n_3!} \nonumber \\
& \quad  = \ell^{n+1} \! \sum_{(n_1, n_2, n_3 ) \in \mathscr{S}} \frac{(n+1)!}{n_1! n_2! n_3!} t_{i-1}^{n_1}(t_i -t_{i-1})^{n_2}(1 - t_{i})^{n_3}\Big|\frac{n_2}{(n+1)(t_i - t_{i-1})} - 1 \Big| \nonumber\\
& \quad = \ell^{n+1} \sum_{k=0}^{n+1} \bigg( \frac{1}{k!} \Big(\sum_{n_1=0}^{k} \frac{k!}{n_1! (k-n_1)!} t_{i-1}^{n_1}(1 - t_{i})^{k-n_1} \Big) \nonumber \\
& \quad  \qquad \qquad \quad  \times \frac{(n+1)!}{(n+1-k)!}(t_i -t_{i-1})^{n+1-k} \Big|\frac{n +1 -k}{(n+1)(t_i - t_{i-1})} - 1 \Big| \bigg) \nonumber \\
& \quad  = \ell^{n+1} \sum_{k=0}^{n+1} \frac{(n+1)!}{k!(n+1-k)!} (1-(t_i-t_{i-1}))^{k}(t_i -t_{i-1})^{n+1-k} \Big|\frac{n +1 -k}{(n+1)(t_i - t_{i-1})} - 1 \Big| \nonumber \\
& \quad  = \ell^{n+1} \! \sum_{k=0}^{n+1}\binom{n+1}{k} (1-(t_i-t_{i-1}))^{n+1-k}(t_i -t_{i-1})^{k} \Big|\frac{k}{(n+1)(t_i - t_{i-1})} - 1 \Big|.
\end{align}
The right-hand sides of \eqref{eq:result_bound_A_1}, \eqref{eq:result_bound_A_3}, and \eqref{eq:result_bound_A_2} correspond to probability mass functions of binomial random variables. Indeed, let
\begin{equation*}
Y_{1,n} \sim \text{Binom}(n, t_{i-1}), \quad Y_{2,n} \sim \text{Binom}(n, t_i - t_{i-1}), \quad \text{and} \quad Y_{3,n} \sim \text{Binom}(n, t_i).
\end{equation*}
Then
\begin{align*}
\| A_1 \| &\leqslant \ell^{n+1} \big( \proba(Y_{1,n} \geqslant p-1) + \proba(Y_{1,n+1} \geqslant p ) \big),\\
 \| A_3 \| & \leqslant \ell^{n+1} \big( \proba(Y_{3,n} \leqslant p-1) + \proba(Y_{3,n+1} \leqslant p-1) \big),\\
\| A_2 \| &\leqslant \ell^{n+1} \esp\Big[\Big|\frac{Y_{2,n+1}}{(n+1)(t_i - t_{i-1})} - 1 \Big|\Big]. 
\end{align*}
First, since $\esp Y_{2,n+1} = (n+1)(t_i-t_{i-1})$ and $\text{Var}(Y_{2,n+1}) = (n+1)(t_i-t_{i-1})(1-(t_i-t_{i-1}))$, by Hölder's inequality,
\begin{align*}
\esp\Big[\Big|\frac{Y_{2,n+1}}{(n+1)(t_i - t_{i-1})} -  1 \Big|\Big] &= \frac{1}{(n+1)(t_i-t_{i-1})}\esp\big[|Y_{2,n+1} - \esp Y_{2,n+1}|\big] \\
& \leqslant \frac{1}{(n+1)(t_i-t_{i-1})}\esp \big[|Y_{2,n+1} - \esp Y_{2,n+1}|^2 \big]^{\nicefrac{1}{2}} \\
& \leqslant \frac{1}{(n+1)(t_i-t_{i-1})} \sqrt{(n+1)(t_i-t_{i-1})(1-(t_i-t_{i-1}))} \\
& \leqslant \frac{1}{\sqrt{n+1}} \sqrt{\frac{1 - (t_i - t_{i-1})}{t_i - t_{i-1}}}.
\end{align*}
The other terms decay exponentially fast if $p$ is well chosen. We give the details for the bound on $\|A_1\|$ below but the one on $\|A_3\|$ is treated in the same way. First, we have
\begin{align*}
\| A_1 \| &\leqslant \ell^{n+1} \big( \proba(Y_{1,n} \geqslant p-1) + \proba(Y_{1,n+1} \geqslant p ) \big),\\
& \leqslant \ell^{n+1} \big( \proba(Y_{1,n} - \esp Y_{1,n} \geqslant p-1 -nt_{i-1}) \big. \\
& \qquad \quad  +\big. \proba(Y_{1,n+1} - \esp Y_{1,n+1} \geqslant p -(n+1)t_{i-1}) \big).
\end{align*}
Recall that $p = \floor{\nicefrac{(n+1)(3t_i + t_{i-1})}{4}}$. Thus, for $n \geqslant \nicefrac{2}{(t_i - t_{i-1})}$, we have
\begin{align*}
\frac{t_i + 3t_{i-1}}{4} \leqslant \frac{p}{n+1} < \frac{3t_i + t_{i-1}}{4}.
\end{align*}
\colorlet{lightgray}{gray!10}
\definecolor{colorp1}{HTML}{0173b2}
\definecolor{colorp2}{HTML}{de8f05}
\definecolor{colorp3}{HTML}{029e73}

\global\def\plotrange{0, 1, 4, 8, 11, 15, 18, 19}
\def\myarr {
    (0.131579, 0.5)
    (0.2631579, 1.228486295)
    (1.05263158, 2.37243038 - 0.3)
    (2.10526316, 2.232905675 - 0.3)
    (2.89473684 - 0.1, 1.597253615)
    (3.94736842 - 0.3, 1.09415549 + 0.2)
    (4.7368421 - 0.5, 1.6927394599999999)
    (5.0 -0.3, 2.5)
}

\global\def\xs{{0.131579, 0.2631579, 0.52631578, 0.78947368, 1.05263158, 1.31578948, 1.57894736, 1.84210526, 2.10526316, 2.36842106, 2.63157894, 2.89473684 - 0.1, 3.15789474, 3.42105264, 3.68421052, 3.94736842 - 0.3, 4.21052632, 4.47368422, 4.7368421 - 0.5, 5.0 - 0.3}}

\global\def\ys{{0.5, 1.228486295, 1.7697186200000001, 2.14419923, 2.37243038 - 0.3, 2.47491434 - 0.3, 2.47215338, 2.384649725, 2.232905675 - 0.3, 2.037423455, 1.81870535, 1.597253615, 1.39357049, 1.2281582599999998, 1.121519165, 1.09415549 + 0.2, 1.166569475, 1.359263375, 1.6927394599999999, 2.5}}

\global\def\xlen{19}

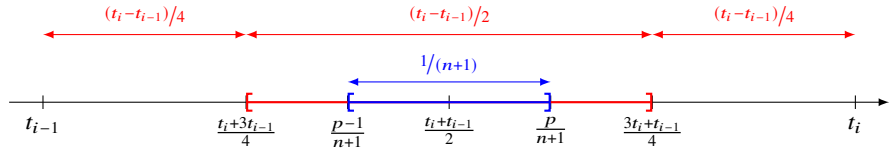
\begin{figure}[ht]
\centering

    \begin{tikzpicture}[scale=0.9]
      \draw[black, ->] (-0.5, 0) -- (12.5, 0);
      
      \draw (0,0.1) -- (0,-0.1) node[below] {$t_{i-1}$};
      \draw (12,0.1) -- (12,-0.1) node[below] {$t_i$};
      \draw (6,0.1) -- (6,-0.1) node[below] {$\frac{t_i + t_{i-1}}{2}$};
      \draw (3,0.1) -- (3,-0.1) node[below] {$\frac{t_i + 3t_{i-1}}{4}$};
      \draw (9,0.1) -- (9,-0.1) node[below] {$\frac{3t_i + t_{i-1}}{4}$};
      \draw (7.5,0.1) -- (7.5,-0.1) node[below] {$\frac{p}{n+1}$};
      \draw (4.5,0.1) -- (4.5,-0.1) node[below] {$\frac{p-1}{n+1}$};
      
      \draw[{[-]}, thick, red] (3,0) -- (9,0);      
      \draw[{[-]}, thick, blue] (4.5,0) -- (7.5,0);

      \draw[blue, <->] (4.5, 0.3) -- (7.5, 0.3) node[pos=0.5,above] {$\nicefrac{1}{(n+1)}$};
      \draw[red, <->] (3, 1) -- (9, 1) node[pos=0.5,above] {$\nicefrac{(t_i - t_{i-1})}{2}$};
      \draw[red, <->] (0, 1) -- (3, 1) node[pos=0.5,above] {$\nicefrac{(t_i - t_{i-1})}{4}$};
      \draw[red, <->] (9, 1) -- (12, 1) node[pos=0.5,above] {$\nicefrac{(t_i - t_{i-1})}{4}$};     
        
    \end{tikzpicture}

\caption{Illustration on the choice of $p$.}
\label{fig:position_p_linear_interval}
\end{figure}
This is summarized in Figure \ref{fig:position_p_linear_interval}. In particular,
\begin{align*}
p -(n+1)t_{i-1} &= (n+1) \Big(\frac{p}{n+1} - t_{i-1}\Big) \\
&\geqslant (n+1) \Big(\frac{t_i + 3t_{i-1}}{4} - t_{i-1}\Big) = (n+1) \frac{t_i - t_{i-1}}{4} > 0,
\end{align*}
and 
\begin{align*}
p -1 - nt_{i-1} = n \Big(\frac{p}{n+1} - t_{i-1} + \frac{p}{n+1} - \frac{p-1}{n}\Big) &\geqslant n \Big(\frac{t_i - t_{i-1}}{4} + \frac{n+1-p}{n(n+1)}\Big) \\
& > n \frac{t_i - t_{i-1}}{4} > 0.
\end{align*}
By Hoeffding's inequality, we obtain
\begin{align*}
\| A_1 \| 
&\leqslant \ell^{n+1}\Big( \exp\Big(-n\Big( \frac{p-1}{n} -t_{i-1}\Big)^2 \Big) +\exp\Big(-(n+1)\Big( \frac{p}{n+1} -t_{i-1}\Big)^2 \Big)\Big) \\
& \leqslant \ell^{n+1} \Big( \exp\Big(- \frac{n(t_i-t_{i-1})^2}{16}\Big) + \exp\Big(-\frac{(n +1)(t_i-t_{i-1})^2}{16}\Big)\Big)\\
& \leqslant 2 \ell^{n+1} \exp\Big(- \frac{n(t_i-t_{i-1})^2}{16}\Big).
\end{align*}
With similar arguments, we obtain
\begin{align*}
\| A_3 \| & \leqslant 2 \ell^{n+1} \exp\Big(- \frac{n(t_i-t_{i-1})^2}{16}\Big).
\end{align*}

Finally, combining \eqref{eq:chen_decomposition_diff_insertion_sig} with the bounds above on $\|A_1\|$, $\|A_2\|$, and $\|A_3\|$, we conclude that
\begin{align*}
&\| \mathscr{L}^n_{p,X} (\beta_i) - (n +1) \mathbf{X}^{n+1}\| \\
& \quad \leqslant \frac{\ell^{n+1}}{n!} \Big(\frac{1}{\sqrt{n+1}} \sqrt{\frac{1 - (t_i - t_{i-1})}{t_i - t_{i-1}}} +  4\exp\Big(- \frac{n(t_i-t_{i-1})^2}{16}\Big) \Big).
\end{align*}

\section{Proof of Proposition \ref{th:sig_lower_bound}}

The proof is based on several ingredients of \citet{hambly2010uniqueness}. Their core idea is to move the path to a hyperbolic space, which is called the ``development'' of the path. First, note that if $Y \in BV(\R^d)$ is defined by $Y_t = X_t/\ell$, where $\ell = \|X\|_{TV}$ is the length of $X$, then 
\begin{equation*}
\mathbf{Y}^n = \int_{(u_1, \dots, u_n) \in \Delta_n} dY_{u_1} \otimes \dots \otimes dY_{u_n} = \int_{(u_1, \dots, u_n) \in \Delta_n} \frac{1}{\ell}dX_{u_1} \otimes \dots \otimes\frac{1}{\ell} dX_{u_n} = \frac{1}{\ell^n} \mathbf{X}^n.
\end{equation*}
Therefore, we will assume without loss of generality that $\ell = 1$, and the lower bound obtained must be multiplied by $\ell^n$ to return to the general case.

We start the proof with a series of definitions.
\begin{definition}
The hyperboloid model is the subspace of $\R^{d+1}$ defined by
\begin{equation*}
\mathbb{H}=\{ y \in \R^{d+1} \mid B(y,y)=-1 \},
\end{equation*}
where, for any $x,y \in \R^{d+1}$, 
\begin{equation*}
\label{eq:def_B}
B(x,y)=\sum_{i=1}^d x_iy_i - x_{d+1}y_{d+1}.
\end{equation*}
\end{definition}

This hyperbolic space has several well-known good properties \citep{cannon1997hyperbolic,paupert2016introduction,loustau2020hyperbolic}. The main one is that the hyperbolic distance between two points $x,y \in \mathbb{H}$, denoted by $d$, can be easily computed as
	\begin{equation*}
		d(x,y)=\text{arcosh}(-B(x, y)).
	\end{equation*}
For $E_1$ and $E_2$ two finite-dimensional normed vector spaces, we let $L(E_1,E_2)$ be the vector space of linear operators from $E_1$ to $E_2$. Equipped with the operator norm, $L(E_1,E_2)$ is itself a normed vector space. Recall that for $f \in L(E_1,E_2)$,
\begin{equation*}
\| f\|_{L(E_1,E_2)} = \underset{x \in E_1, \|x\|_{E_1}=1}{\sup} \| f(x) \|_{E_2}.
\end{equation*}

We will use linear maps on $\R^d$ equipped with the Euclidean norm, and consider linear maps as matrix multiplications. Therefore, we write $fx$ instead of $f(x)$ for function evaluation. In particular, we let $F:\R^{d} \to L(\R^{d+1}, \R^{d+1})$ be defined for any $y \in \R^d$ by
\begin{equation*}
	Fy= \begin{pmatrix} 
	0 & \cdots & 0 &  y_1 \\ 
	\vdots & & \vdots &\vdots \\
	0 & \cdots & 0 & y_d \\
	y_1 & \dots & y_d & 0
	 \end{pmatrix}.
\end{equation*}

\begin{lemma}
\label{Fnorm}
$F$ is a bounded linear map and its operator norm is 
\[\|F\|_{L(\R^{d}, L(\R^{d+1}, \R^{d+1}))}=1.\]
\end{lemma}
\begin{proof}
For any $y \in \R^d$, $z \in \R^{d+1}$, we have
\begin{align*}
Fyz= \begin{pmatrix} 
	0 & \cdots & 0 &  y_1 \\ 
	\vdots & & \vdots &\vdots \\
	0 & \cdots & 0 & y_d \\
	y_1 & \dots & y_d & 0
	 \end{pmatrix} \begin{pmatrix} z_1 \\ \vdots \\ z_{d+1} \end{pmatrix} = \begin{pmatrix}y_1\, z_{d+1} \\ \vdots \\ y_d \, z_{d+1} \\ \sum_{i=1}^d y_i \, z_i \end{pmatrix}.
\end{align*}
Thus,
\begin{align*}
 \|Fyz\|^2 &=  z_{d+1}^2 \sum_{i=1}^d y_i^2 + \big(\sum_{i=1}^d y_i z_i \big)^2 \\
 	& \leqslant z_{d+1}^2 \|y\|^2  + \|y\|^2 \sum_{i=1}^d z_i^2 &\\
 	& \quad \text{(by the Cauchy-Schwartz inequality)} \\
 	& \leqslant \|y\|^2 \|z\|^2.
\end{align*}
This yields
\begin{align*}
\| Fy\|_{L(\R^{d+1}, \R^{d+1})} &= \underset{z \in \R^{d+1}, \|z\|=1}{\sup} \|Fyz\| \leqslant \underset{z \in \R^{d+1}, \|z\|=1}{\sup} \|y\| \, \|z\| = \|y\|.
\end{align*}
This inequality becomes an equality when $z=e_{d+1}$ (where $(e_1, \dots, e_{d+1})$ is the canonical basis of $\R^{d+1}$). Therefore,
\begin{align*} 
\| Fy\|_{L(\R^{d+1}, \R^{d+1})}=\|y\|,
\end{align*}
and 
\begin{equation*}
\|F\|_{L(\R^{d}, L(\R^{d+1}, \R^{d+1}))} = \sup_{y \in \R^d, \|y\|=1} \| Fy\|_{L(\R^{d+1}, \R^{d+1})}=\sup_{y \in \R^d, \|y\|=1} \|y\| = 1.
\end{equation*}
{\qed}
\end{proof}

We are now ready to define the development of a path $X \in BV(\R^d)$ to the hyperbolic space. For $t \in [0,1]$, we let $\Gamma_t : \R^{d+1} \to \R^{d+1}$ be such that, for any $y \in \R^{d+1}$, $\Gamma_t y$ is the solution at time $t$ of the controlled differential equation
\begin{equation}
\label{eq:cde_Y_def}
dY_t = F(dX_t)Y_t, \quad Y_0=y.
\end{equation}

\begin{lemma}
\label{lemma:Gamma_t_linear_isometry}
For any $t \in [0,1]$, $\Gamma_t$ is a (well-defined) linear map preserving the function $B$, i.e., for any $y,\widetilde{y} \in \R^{d+1}$,
\[B(\Gamma_ty, \Gamma_t \widetilde{y})= B(y, \widetilde{y}). \]
\end{lemma}
\begin{proof}
Equation \eqref{eq:cde_Y_def} is a linear controlled differential equation and, by Picard's theorem \citep[see, e.g.,][Theorem 1.3]{lyons2007differential}, for any $y \in \R^d$, it has a unique solution. Therefore $\Gamma_t$ is well-defined. Moreover, it is clearly a linear operator: if $y,\widetilde{y} \in \R^d$, and $Y_t=\Gamma_t y$, $\widetilde{Y}_t = \Gamma_t \widetilde{y}$, then $Y_t + \widetilde{Y}_t$ follows equation \eqref{eq:cde_Y_def} with initial point $y+\widetilde{y}$, so by uniqueness of the solution, $\Gamma_t(y+\widetilde{y})= \Gamma_ty + \Gamma_t \widetilde{y}$. 

The fact that $\Gamma$ preserves $B$ is proved in \citet{hambly2010uniqueness}, Section 3.
{\qed}
\end{proof}

The differential equation \eqref{eq:cde_Y_def} may be rewritten as a differential equation on $\Gamma_t$, which itself can be thought of as a $(d+1)\times (d+1)$ matrix:
\begin{equation} \label{eq:cde_Gamma_def}
d\Gamma_t = F(dX_t) \Gamma_t, \quad \Gamma_0=I_{d+1},
\end{equation} 
where $I_{d+1}$ is the identity matrix. Now, let $y_0=(0, \dots,0,1)^\top \in\R^{d+1}$. The development of $X$ to  the hyperboloid space $\mathbb{H}$ is defined as the path $Y: [0,1] \to \mathbb{H}$, $Y_t = \Gamma_t y_0$. Observe, since $B(y_0,y_0)=-1$, that $y_0 \in \mathbb{H}$ and thus, by Lemma \ref{lemma:Gamma_t_linear_isometry}, that $Y_t \in \mathbb{H}$ for any $t \in [0,1]$. To summarize, equation \eqref{eq:cde_Gamma_def} maps a path $X$ in $\R^d$ into a new path $Y$ in $\mathbb{H}$. It turns out that this embedding  preserves a number of properties of the path $X$. In particular \citep{hambly2010uniqueness}:
\begin{itemize}
	\item Any linear piece of $X$ is mapped to a geodesic in $\mathbb{H}$. So, if $X$ is piecewise linear, then $Y$ is geodesic on each $[t_{i-1},t_i]$, and $d(Y_{t_i},Y_{t_{i-1}}) =\| X_{t_i} - X_{t_{i-1}}\|$. Therefore, since $X$ is linear on $[t_{i-1}, t_i]$, $d(Y_{t_i},Y_{t_{i-1}})=\|\beta_i\|(t_i - t_{i-1})$.
	\item The function $t \mapsto \Gamma_t$ preserves the angles between linear segments. If $2\omega$ is the angle between the linear pieces $[X_{t_{i-1}},X_{t_i}]$ and $[X_{t_i}, X_{t_{i+1}}]$, then the angle between the geodesics $[Y_{t_{i-1}},Y_{t_i}]$ and $[Y_{t_i}, Y_{t_{i+1}}]$ is also equal to $2\omega$.
\end{itemize}

The last ingredient to prove Proposition \ref{th:sig_lower_bound} is the following lemma, due to \citet[][Lemma 3.7 and Proposition 3.13]{hambly2010uniqueness}. 
\begin{lemma} 
\label{lemma:results_hambly_lyons}
$ $
\begin{itemize}
\item[$(i)$] Let $Y:[0,1] \to \mathbb{H}$ be a continuous path, geodesic on the intervals $[t_{i-1},t_i]$, where $0=t_0< t_1 < \dots < t_M=1$ is a partition of $[0,1]$. If $\omega$ is the smallest angle between two geodesic segments and each geodesic segment has length at least $K(\omega)$, then 
\begin{equation*}
0 \leqslant \sum_{i=1}^M d(Y_{t_{i-1}}, Y_{t_i}) - d(y_0, Y_{1}) \leqslant (M-1) K(\omega).
\end{equation*}
\item[$(ii)$] Let $SO(B)$ denote the group of $(d+1) \times (d+1)$ matrices with positive determinant preserving the isometry $B$: if $G \in SO(B)$, for any $x,y \in \R^{d+1}$, $B(Gx,Gy)=B(x,y)$. Then, for any $G \in SO(B)$, 
\[ \|G\|_{L(\R^{d+1}, \R^{d+1})} \geqslant e^{d(y_0, Gy_0)},\]
where $y_0=(0, \dots, 0, 1)^\top \in \R^{d+1}$.
\end{itemize}
\end{lemma}

We now have all the components necessary to prove Proposition \ref{th:sig_lower_bound}. Recall that it is assumed, without loss of generality, that $X$ has total variation 1. For $\alpha>0$, we let $X^\alpha_t= \alpha X_t$, and denote by $\Gamma^\alpha_t$ and $Y^\alpha_t$ the corresponding hyperbolic developments. Let $D$ be the length of the shortest linear segment of $X$, and take $\alpha$ such that $\alpha > K(\omega)/D$. Clearly, the path $t \mapsto Y^\alpha_t$ satisfies the assumptions of Lemma \ref{lemma:results_hambly_lyons}, $(i)$. Therefore
\begin{align*}
0 \leqslant \sum_{i=1}^M d(Y^\alpha_{t_{i-1}}, Y^\alpha_{t_i}) - d(y_0, Y^\alpha_1) \leqslant (M-1) K(\omega).
\end{align*}
Since $X$ is linear on each segment $[t_{i-1}, t_i]$, $d(Y^\alpha_{t_{i-1}}, Y^\alpha_{t_i})= \| X^\alpha_{t_i} - X^\alpha_{t_{i-1}} \|$, and thus
\[ \sum_{i=1}^M d(Y^\alpha_{t_{i-1}}, Y^\alpha_{t_i}) = \sum_{i=1}^M \| X_{t_i}^\alpha - X_{t_{i-1}}^\alpha \| = \alpha.\]
Hence,
\[d(y_0, Y^\alpha_1) \geqslant \alpha - (M-1)K(\omega). \]
Therefore, upon noting that $ d(y_0, \Gamma_1^\alpha y_0) = d(y_0, Y_1^\alpha)$, according to Lemma \ref{lemma:results_hambly_lyons},
\begin{align} \label{eq:lower_bound_Gamma_alpha}
\| \Gamma_1^\alpha \|_{L(\R^{d+1}, \R^{d+1})} \geqslant e^{d(y_0, \Gamma_1^\alpha y_0)} \geqslant e^{\alpha - (M-1) K(\omega)}.
\end{align}

Let us now prove that $\Gamma_1$ is in fact a linear function of the signature of $X$. To see this, observe, since $\Gamma_1$ is the solution at time $1$ to the linear controlled differential equation \eqref{eq:cde_Gamma_def}, that
\begin{align*}
\Gamma_1^\alpha &= I_{d+1} + \int_{0}^1 F(dX^\alpha_t)\Gamma_t \\
& = I_{d+1} + \int_{0}^1 F(dX^\alpha_t) \big( I_{d+1} + \int_{0}^t F(dX^\alpha_s) \Gamma_s\big) \\
& =I_{d+1} + \int_{0}^1 F(dX^\alpha_t) + \int_0^1 \int_0^t F(dX^\alpha_t) F(dX^\alpha_s) \Gamma_s\\
& = \dots
\end{align*}
Iterating this procedure yields the identity
\begin{equation} \label{eq:Gamma_infinite_series}
\Gamma_1^\alpha = I_{d+1} + \sum_{k=1}^\infty \int_{\Delta_k} F(dX^\alpha_{t_1}) \cdots F(dX^\alpha_{t_k}).
\end{equation}
In order to show that the series above is well-defined, denote by $F^{\otimes k}$ the linear map $F^{\otimes k} : (\R^d)^{\otimes k} \to L(\R^{d+1}, \R^{d+1})$ such that, for any tensor $u = \sum_{j=1}^\ell u_{1,j} \otimes \dots \otimes u_{k,j}$,
\[ F^{\otimes k}(u) = \sum_{j=1}^\ell F(u_{1,j}) \cdots F(u_{k,j}).\]
Thus, 
\begin{align*}
\| F^{\otimes k}(u)\|_{L(\R^{d+1}, \R^{d+1})} &\leqslant \sum_{j=1}^p \|F(u_{1,j}) \cdots F(u_{k,j}) \|_{L(\R^{d+1}, \R^{d+1})} \\
& \leqslant \sum_{j=1}^p \|F(u_{1,j})\|_{L(\R^{d+1}, \R^{d+1})} \cdots \|F(u_{k,j})\|_{L(\R^{d+1}, \R^{d+1})} \\
& \qquad \text{(because the operator norm is sub-multiplicative)} \\
& \leqslant  \sum_{j=1}^p \|u_{1,j}\|\cdots \|u_{k,j}\| \\
& \qquad \text{(by Lemma \ref{Fnorm}).}
\end{align*}
Taking the infimum over any representation of $u$ yields $\| F^{\otimes k}(u)\|_{L(\R^{d+1}, \R^{d+1})} \leqslant \|u\|_\pi$. (Note that this is not true for any tensor norm.)
Endowing $(\R^d)^{\otimes k}$ with the projective norm, the $k$th term in the sum \eqref{eq:Gamma_infinite_series} becomes
\begin{align*}
F(dX^\alpha_{t_1}) \cdots F(dX^\alpha_{t_k}) &= F^{\otimes k} \Big( \int_{\Delta_k} dX^\alpha_{t_1} \otimes \dots \otimes dX^\alpha_{t_k} \Big) \\
&= \alpha^k F^{\otimes k} \Big( \int_{\Delta_k} dX_{t_1} \otimes \dots \otimes dX_{t_k} \Big) \\
& =\alpha^k F^{\otimes k} (\mathbf{X}^k),
\end{align*}
and
\begin{align*}
\big\| \alpha^k F^{\otimes k} (\mathbf{X}^k) \big\|_{L(\R^{d+1}, \R^{d+1})} = \alpha^k \big\|  F^{\otimes k} (\mathbf{X}^k) \big\|_{L(\R^{d+1}, \R^{d+1})} \leqslant \alpha^k \| \mathbf{X}^k \|_\pi \leqslant \alpha^k / k!
\end{align*}
by Proposition \ref{prop:up_bound_norm_sig}. This shows that the right-hand side of \eqref{eq:Gamma_infinite_series} is a convergent series. In addition, we have just proved that
\begin{align*}
\Gamma_1^\alpha = I_{d+1} + \sum_{k=1}^\infty \alpha^k F^{\otimes k} (\mathbf{X}^k),
\end{align*}
i.e., $\Gamma_1^\alpha$ is a linear function of the signature, and, moreover, that
\begin{align} \label{eq:bound_Gamma_1}
\| \Gamma_1^\alpha\|_{L(\R^{d+1}, \R^{d+1})} 
& \leqslant 1 + \sum_{k=1}^\infty \alpha^k \| \mathbf{X}^k \|_\pi .
\end{align}

Next, let, for any $k \geqslant 1$, $b_k = k! \| \mathbf{X}^k \|_\pi$, and $b_0=1$. Note that $ b_k \leqslant 1$ since we have taken $X$ of total variation 1. Thus, combining  \eqref{eq:lower_bound_Gamma_alpha} with \eqref{eq:bound_Gamma_1} yields
\begin{equation} \label{eq:lower_bound_esp_b_Z}
e^{- (M-1)K(\omega)} \leqslant e^{-\alpha}\sum_{k=0}^\infty \frac{\alpha^k}{k!} b_k.
\end{equation}

The last step of the proof is to use inequality \eqref{eq:lower_bound_esp_b_Z} to prove that for each $n > n_1$ there exists an integer $k_n \in [n - n^{3/4}, n + n^{3/4}]$ such that $b_{k_n} \geqslant \nicefrac{e^{-(M-1)K(\omega)}}{2}$. Indeed, if such a sequence $(k_n)$ exists, then
\[\|\mathbf{X}^{k_n} \|_{\pi} = \frac{b_{k_n}}{k_n!} \geqslant \frac{e^{-(M-1)K(\omega)}}{2 k_n!},\]
which is the desired conclusion.

Let $Z \sim \text{Poisson}(\alpha)$ be a random variable following a Poisson distribution with parameter $\alpha$. The right-hand-side of \eqref{eq:lower_bound_esp_b_Z} is then exactly $\esp[ b_Z]$. In the next few lines, we show that 
\begin{equation} \label{eq:proba_b_Z}
\proba \big( b_Z \geqslant \frac{1}{2} e^{-(M-1)K(\omega)} \text{ and } Z \in J_{n} \big) > 0,
\end{equation}
where $J_n$ denotes the interval $[n-n^{3/4}, n + n^{3/4}]$. It follows that there exists $k_n \in J_n$ such that $b_{k_n} \geqslant \frac{1}{2} e^{-(M-1)K(\omega)}$ (otherwise the probability above would be zero). In order to prove \eqref{eq:proba_b_Z}, observe, since $b_k \leqslant 1$ for all $k\geqslant 0$, that
\begin{align*}
e^{-(M-1)K(\omega)} \leqslant \esp[b_Z] &= \esp[b_Z \ind_{b_Z \geqslant \frac{1}{2} e^{-(M-1)K(\omega)}}] + \esp[b_Z \ind_{b_Z < \frac{1}{2} e^{-(M-1)K(\omega)}}] \\
& \leqslant \proba(b_Z \geqslant \frac{1}{2} e^{-(M-1)K(\omega)}) + \frac{1}{2} e^{-(M-1)K(\omega)}.
\end{align*}
This yields the inequality
\begin{align*}
\proba(b_Z \geqslant \frac{1}{2} e^{-(M-1)K(\omega)}) \geqslant \frac{1}{2} e^{-(M-1)K(\omega)} >0.
\end{align*}
Moreover, by Chebyshev's inequality,
\begin{align*}
\proba(b_Z \notin J_n) = \proba( |Z-n| > n^{3/4}) \leqslant \frac{1}{\sqrt{n}}.
\end{align*}
Let $n_1 =  \lfloor 4 e^{2(M-1)K(\omega)} \rfloor$. Then, for any $n > n_1$,
\begin{align*}
\proba \big( b_Z \geqslant \frac{1}{2} e^{-(M-1)K(\omega)} \text{ and } Z \in J_{n} \big) &\geqslant \proba(b_Z \geqslant \frac{1}{2} e^{-(M-1)K(\omega)}) - \proba(b_Z \notin J_n) \\
& \geqslant \frac{1}{2} e^{-(M-1)K(\omega)} - \frac{1}{\sqrt{n}} \\
& \geqslant \frac{1}{2} e^{-(M-1)K(\omega)} - \frac{1}{\sqrt{n_1}} >0.
\end{align*}
This shows inequality \eqref{eq:proba_b_Z}.

\section{Proof of Lemma \ref{lemma:singular_values_insertion_operator}}

Recall that $(e_1, \dots, e_d)$ denotes the canonical basis of $\R^d$. Each column of $A_p$ corresponds to $\mathscr{L}^n_{p,X}(e_j)$, and, for any $(n+1)$-tuple $(i_1,\dots,i_{n+1}) \in \{1, \dots, d\}^{n+1}$, one has
\begin{equation*}
(\mathscr{L}^n_{p,X}(e_j))_{i_1, \dots, i_{n+1}} = \begin{cases} S^{(i_1,\dots, i_{p-1},i_{p+1}, \dots, i_{n+1})} \text{ if } i_p=j, \\ 0 \text{ otherwise. } \end{cases}
\end{equation*}
Therefore, each row of $A_p$, indexed by an $(n+1)$-tuple $(i_1,\dots, i_{n+1})$, contains only one non-zero element, namely the one in the column $j=i_p$. Let $a_{i_1,\dots, i_{n+1},j}$ denote the element in column $j$ corresponding to the row $(i_1,\dots,i_{n+1})$. Then, for any $q,r \in \{1,\dots, d\}$,
\begin{align*}
(A_p^\top A_p)_{q,r}&= \sum_{(i_1,\dots,i_{n+1}) \in \{1, \dots, d\}^{n+1}} a_{i_1,\dots, i_{n+1},q} \, a_{i_1,\dots, i_{n+1},r}\\
&= \sum_{i_p=1}^d \sum_{(i_1,\dots, i_{p-1},i_{p+1},\dots, i_{n+1}) \in \{1, \dots, d\}^{n}}a_{i_1,\dots, i_{n+1},q} \, a_{i_1,\dots, i_{n+1},r} 
\end{align*}
For both $a_{i_1,\dots, i_{n+1},q}$ and $a_{i_1,\dots, i_{n+1},r}$ to be non-zero, $q$ and $r$ must be equal to $i_p$. Therefore, $A_p^\top A_p$ is diagonal. The diagonal terms, corresponding to $i_p=q=r$, are equal to
\begin{align*}
(A_p^\top A_p)_{q,q}& = 
\sum_{(i_1,\dots, i_{p-1},i_{p+1},\ldots i_{n+1}) \in \{1, \dots, d\}^{n}} (S^{(i_1,\dots, i_{p-1},i_{p+1}, \dots, i_{n+1})})^2=\| \mathbf{X}^n \|^2.
\end{align*}

\end{document}